\documentclass[final]{dmtcs-episciences}

\usepackage[utf8]{inputenc}
\usepackage{subfigure}
\usepackage[round]{natbib}

\usepackage{amssymb}
\usepackage{amsmath}
\usepackage{latexsym}
\usepackage{verbatim}
\usepackage{amsfonts}
\usepackage{latexsym}
\usepackage{xcolor}
\usepackage{url}
\usepackage{algorithm}
\usepackage{algorithmic}

\newcommand{\per}{{per}}

\newcommand{\bigo}{{\mathcal O}}

\newcommand{\ra}{\rightarrow}

\newcommand{\LCP}{{\mathit{LCP}}}

\newcommand{\LPF}{{\mathit{LPrF}}}
\newcommand{\LPdF}{{\mathit{LPF}}}
\newcommand{\LP}{{\mathit{LPal}}_\alpha}
\newcommand{\LR}{{\mathit{LRep}}_\alpha}

\newtheorem{remark}{Remark}
\newtheorem{lemma}{Lemma}
\newtheorem{problem}{Problem}
\newtheorem{theorem}{Theorem}

\author{Marius Dumitran\affiliationmark{1} \and Pawe\l{} Gawrychowski\affiliationmark{2} \and Florin Manea\affiliationmark{3}}
\title[Longest Gapped Repeats and Palindromes]{Longest Gapped Repeats and Palindromes
\thanks{This is an extension of the conference papers of \cite{fct}, presented at the 20th International Symposium on
Fundamentals of Computation Theory, FCT 2015, and of \cite{mfcs}, presented at the 40th International Symposium on
Mathematical Foundations of Computer Science, MFCS 2015.}}
\affiliation{
  Faculty of Mathematics and Computer Science, University of Bucharest \\
  Institute of Computer Science, University of Wroc\l{}aw\\
  Department of Computer Science, Kiel University}
\received{2015-11-24}
\revised{2016-10-12, 2017-9-25}
\accepted{2017-10-3}

\begin{document}
\publicationdetails{19}{2017}{4}{4}{1337}
\maketitle
\begin{abstract}
A gapped repeat (respectively, palindrome) occurring in a word $w$ is a factor $uvu$ (respectively, $u^Rvu$) of $w$. In such a repeat (palindrome) $u$ is called the arm of the repeat (respectively, palindrome), while $v$ is called the gap. 
We show how to compute efficiently, for every position $i$ of the word $w$, the longest gapped repeat and palindrome occurring at that position, provided that the length of the gap is subject to various types of restrictions. That is, that for each position $i$ we compute the longest prefix $u$ of $w[i..n]$ such that $uv$ (respectively, $u^Rv$) is a suffix of $w[1..i-1]$ (defining thus a gapped repeat $uvu$ -- respectively, palindrome $u^Rvu$), and the length of $v$ is subject to the aforementioned restrictions.
\end{abstract}

\section{Introduction}
Gapped repeats and palindromes have been investigated for a long time (see, e.g., \cite{Gu97,Brodal,KK_SPIRE,KK09,KolpakovPPK14,power_of_SA,Cro2011} and the references therein), with motivation coming especially from the analysis of DNA and RNA structures, where tandem repeats or hairpin structures play important roles in revealing structural and functional information of the analysed genetic sequence (see \cite{Gu97,Brodal,KK09} and the references therein). 

A gapped repeat (respectively, palindrome) occurring in a word $w$ is a factor $uvu$ (respectively, $u^Rvu$) of $w$. The middle part $v$ of such a structure is called gap, while the two factors $u$ (respectively, the factors $u^R$ and $u$) are called left and right arms. Generally, the previous works were interested in finding all the gapped repeats and palindromes, under certain numerical restrictions on the length of the gap or on the relation between the length of the arm of the repeat or palindrome and the length of the gap.

A classical problem for palindromes asks to find the longest palindromic factor of a word (see, \cite{Manacher}). This is our first inspiration point in proposing an alternative point of view in the study of gapped repeats and palindromes. As a second inspiration point for our work, we refer to the longest previous factor table (LPF) associated to a word. This data structure was introduced and considered in the context of efficiently computing Lempel-Ziv-like factorisations of words (see \cite{IlieLPF,power_of_SA}). Such a table provides for each position $i$ of the word the longest factor occurring both at position $i$ and once again on a position $j<i$. Several variants of this table were also considered by \cite{power_of_SA}: the longest previous reverse factor ($\LPF$), where we look for the longest factor occurring at position $i$ and whose mirror image occurs in the prefix of $w$ of length $i-1$, or the longest previous non-overlapping factor, where we look for the longest factor occurring both at position $i$ and somewhere inside the prefix of length $i-1$ of $w$. Such tables may be seen as providing a comprehensive image of the long repeats and symmetries occurring in the analysed word. 

According to the above, in our work we approach the construction of longest previous gapped repeat or palindrome tables: for each position $i$ of the word we want to compute the longest factor occurring both at position $i$ and once again on a position $j<i$ (or, respectively, whose mirror image occurs in the prefix of length $i-1$ of $w$) such that there is a gap (subject to various restrictions) between $i$ and the previous occurrence of the respective factor (mirrored factor). Similar to the original setting, this should give us a good image of the long gapped repeats and symmetries of a~word.

A simple way to restrict the gap is to lower bound it by a constant; i.e., we look for factors $uvu$ (or $u^Rvu$) with $|v|>g$ for some $g\geq 0$. The techniques of~\cite{power_of_SA} can be easily adapted to compute for each position $i$ the longest prefix $u$ of $w[i..n]$ such that there exists a suffix $uv$ (respectively, $u^Rv$) of $w[1..i-1]$, forming thus a factor $uvu$ (respectively, $u^Rvu$) with $|v|>g$. Here we consider three other different types of restricted gaps.

We first consider the case when the length of the gap is between a lower bound $g$ and an upper bound $G$, where $g$ and $G$ are given as input (so, may depend on the input word). This extends naturally the case of lower bounded~gaps. 
\begin{problem}\label{LPFgG}
Given $w$ of length $n$ and two integers $g$ and $G$, such that $0\leq g< G\leq n$, construct the arrays $\LPF_{g,G}[\cdot]$ and $\LPdF_{g,G}[\cdot]$ defined for $1\leq i\leq n$:
\begin{itemize} 
\item[a.] $\LPF_{g,G}[i]=\max\{|u|\mid$ there exists $v$ such that $u^Rv$ is a suffix of $w[1..i-1]$ and $u$ is prefix of $w[i..n]$, with $ g\leq  |v|< G\}$.
\item[b.]$\LPdF_{g,G}[i]=\max\{|u|\mid$ there exists $v$ such that $uv$ is a suffix of $w[1..i-1]$ and $u$ is prefix of $w[i..n]$, with $ g\leq |v|< G\}$.
\end{itemize}
\end{problem}
We are able to solve Problem \ref{LPFgG}(a) in linear time $\bigo(n)$. Problem \ref{LPFgG}(b)  is solved here in $\bigo(n\log n)$ time. Intuitively, in the case of gapped palindromes, when trying to compute the longest prefix $u$ of $w[i..n]$ such that $u^Rv$ is a suffix of $w[1..i-1]$ with $g<|v|\leq G$, we just have to compute the longest common prefix between $w[i..n]$ and the words $w[1..j]^R$ with $g<i-j\leq G$. The increased difficulty in solving the problem for repeats (reflected in the increased complexity of our algorithm) seems to come from the fact that when trying to compute the longest prefix $u$ of $w[i..n]$ such that $uv$ is a suffix of $w[1..i-1]$ with $g<|v|\leq G$, it is hard to see where the $uv$ factor may start, so we have to somehow try more variants for the length of $u$. \cite{Brodal} give an algorithm that finds all maximal repeats (i.e., repeats whose arms cannot be extended) with gap between a lower and an upper bound, running in $\bigo(n\log n+z)$ time, where $z$ is the number of such repeats. It is worth noting that there are words (e.g., $(a^2b)^{n/3}$, from \cite{Brodal}) that may have $\Theta(n G)$ maximal repeats $uvu$ with $|v|<G$. Hence, for $G>\log n$ and $g=0$, for instance, our algorithm is faster than an approach that would first use the algorithms of \cite{Brodal} to get all maximal repeats, and then process them somehow to solve Problem \ref{LPFgG}(b). 

The data structures we construct allow us to trivially find in linear time the longest gapped palindrome having the length of the gap between $g$ and $G$, and in $\bigo(n \log n)$ time the longest gapped repeat with the length of the gap between the bounds $g$ and $G$.  

In the second case, the gaps of the repeats and palindromes we investigate are only lower bounded; however, the bound on the gap allowed at each position is defined by a function depending on the~position.
\begin{problem}\label{LPFg(i)}
Given $w$ of length $n$ and the values $g(1),\ldots,g(n)$ of $g:\{1,\ldots,n\}\ra \{1,\ldots,n\}$, construct the arrays $\LPF_{g}[\cdot]$ and $\LPdF_{g}[\cdot]$ defined for $1\leq i\leq n$:
\begin{itemize}
\item[a.]$\LPF_{g}[i]=\max\{|u|\mid$  there exists $v$ such that $u^Rv$ is a suffix of $w[1..i-1]$ and $u$ is prefix of $w[i..n]$, with $g(i)\leq |v|\}$.
\item[b.]$\LPdF_{g}[i]=\max\{|u|\mid$  there exists $v$ such that $uv$ is a suffix of $w[1..i-1]$ and $u$ is prefix of $w[i..n]$, with $g(i)\leq |v|\}$.
\end{itemize}
\end{problem}
The setting of this problem can be seen as follows. An expert preprocesses the input word (in a way specific to the framework in which one needs this problem solved), and detects the length of the gap occurring at each position (so, computes $g(i)$ for all $i$). These values and the word are then given to us, to compute the arrays defined in our problems. We solve both problems in linear time. Consequently, we can find in linear time the longest gapped palindrome or repeat whose gap fulfils the length restriction defined by the position where this palindrome or repeat occurs (as above).

Finally, following \cite{KK09,KolpakovPPK14}, we analyse gapped repeats $uvu$ or palindromes $u^Rvu$ where the length of the gap $v$ is upper bounded by the length of the arm $u$ multiplied by some factor. More precisely, \cite{KolpakovPPK14} investigate {\em $\alpha$-gapped repeats}: words $uvu$ with $|uv|\leq \alpha|u|$. Similarly, \cite{KK09} analyse {\em $\alpha$-gapped palindromes}, i.e., words $u^Rvu$ with $|uv|\leq \alpha|u|$.  For $\alpha=2$, these structures are called long armed repeats (or pairs) and palindromes, respectively; for $\alpha=1$, they are squares and palindromes of even length, respectively. Intuitively, one is interested in repeats or palindromes whose arms are roughly close one to the other; therefore, the study of $\alpha$-gapped repeats and palindromes was rather focused on the cases with small $\alpha$. Here, we address the general case, of searching in a word $w$ $\alpha$-gapped repeats or palindromes for $\alpha\leq |w|$.

\begin{problem}\label{LLAP}
Given  $w$ of length $n$ and a number $\alpha\leq n$, construct the arrays $\LP[\cdot]$ and $\LR[\cdot]$, defined for $1\leq i\leq n$:
\begin{itemize}
\item[a.]$\LP[i]=\max\{|u|\mid$ there exists $v$ such that $u^Rv$ is a suffix of $w[1..i-1], u$ is a prefix of $w[i..n]$, and $|uv|\leq \alpha |u|\}$.
\item[b.]$\LR[i]=\max\{|u|\mid$ there exists $v$ such that $uv$ is a suffix of $w[1..i-1], u$ is a prefix of $w[i..n]$, and and $|uv|\leq \alpha |u|\}$.
\end{itemize}
\end{problem}

The problem of constructing the set $S$ of all factors of a word of length $n$ which are maximal $\alpha$-gapped repeats of palindromes (i.e., the arms cannot be extended simultaneously with one symbol to the right or to the left to get a longer similar structure) was thoroughly considered, and finally settled by \cite{CroKolKu2015,STACS2016} (see also the references therein). In both these papers, it is shown that the number of $\alpha$-gapped repeats or palindromes a word of length $n$ may contain is $\Theta(\alpha n)$. Using as starting point the algorithm presented in \cite{fct}, that finds the longest $\alpha$-gapped repeat or palindrome (without constructing the set of all such structures), \cite{STACS2016} give an algorithm finding all maximal $\alpha$-gapped repeats and $\alpha$-gapped palindromes in optimal time $\Theta(\alpha n)$.  Here, we first present the algorithm of \cite{fct} for the identification of the longest $\alpha$-gapped repeat or palindrome contained in a word, and briefly explain how it was extended to output all maximal $\alpha$-gapped repeats and palindromes in a word. Then we use the algorithm of \cite{STACS2016} and a linear time algorithm finding the longest square/palindrome centred at each position of a word to solve Problem~\ref{LLAP} in linear time. 

Our algorithms are generally based on efficient data-structures. On one hand, we use efficient word-processing data structures like suffix arrays, longest common prefix structures, or dictionaries of basic factors. On the other hand, we heavily use specific data-structures for maintaining efficiently collections of disjoint sets, under union and find ope\-rations. Alongside these data-structures, we make use of a series of remarks of combinatorial nature, providing insight in the repetitive structure of the words. 

\section{Preliminaries}
The computational model we use to design and analyze our algorithms is the standard unit-cost RAM with logarithmic word size, which is generally used in the analysis of algorithms. In this model, the memory word size is logarithmic in the size of the input. 

Let $V$ be a finite alphabet; $V^*$ denotes the set of all finite words over $V$. In the upcoming algorithmic problems, we assume that the words we process are sequences of integers (i.e., over integer alphabets). In general, if the input word has length $n$ then we assume its letters are in $\{1,\ldots,n\}$, so each letter fits in a single memory-word. This is a common assumption in stringology (see, e.g., the discussion by \cite{KaSaBu06}).

The \emph{length} of a word $w\in V^*$ is denoted by $\left|w\right|$. The \emph{empty word} is denoted by ${\lambda}$. 
A word $u\in V^*$ is a \emph{factor} of $v\in V^*$ if $v=xuy$, for some $x, y\in V^*$; we say that $u$ is a \emph{prefix} of $v$, if $x={\lambda}$, and a \emph{suffix} of $v$, if $y={\lambda}.$
We denote by $w[i]$ the symbol occurring at position $i$ in $w,$ and by $w[i..j]$ the factor of $w$ starting at position $i$ and ending at position $j,$ consisting of the catenation of the symbols $w[i], \ldots, w[j],$ where $1\leq i\leq j\leq n$; we define $w[i..j]=\lambda$ if $i>j$. The powers of a word $w$ are defined recursively by $w^0={\lambda}$ and $w^n=ww^{n-1}$  for $n\geq1$. 
If $w$ cannot be expressed as a nontrivial power (i.e., $w$ is not a repetition) of another word, then $w$ is \emph{primitive}. 
A \emph{period} of a word $w$ over $V$ is a positive integer $p$ such that $w[i]=w[j]$ for all $i$ and $j$ with $i\equiv j\pmod{p}$. Let $\per(w)$ be the smallest~period~of~$w$. A word $w$ with $\per(w)\leq\frac{|w|}{2}$ is called periodic; a periodic $w[i..j]$ (with $p=\per(w[i..j])<\frac{j-i+1}{2}$) is a run if it cannot be extended to the left or right to get a word with the same period $p$, i.e., $i = 1$ or $w[i-1] \neq w[i+p-1]$,
and, $j = n$ or $w[j + 1] \neq w[j - p + 1]$. \cite{KK99} showed that the number of runs of a word is linear and their list (with a run $w[i..j]$ represented as the triple $(i,j,\per(w[i..j])$) can be computed in linear time.
The exponent of a run $w[i..j]$ occurring in $w$ is defined as $\frac{j-i+1}{\per(w[i..j])}$; the sum of the exponents of all runs in a word of length $n$ is $\bigo(n)$ (see \cite{KK99}). 

For a word $u$, $|u|=n$, over $V\subseteq \{1,\ldots,n\}$ we build in $\bigo(n)$ time the suffix array as well as data structures allowing us to retrieve in constant time the length of the longest common prefix of any two suffixes $u[i..n]$ and $u[j..n]$
of $u$, denoted $\LCP_u(i,j)$ (the subscript $u$ is omitted when there is no danger of confusion). Such structures are called $\LCP$ data structures in the following. For details, see, e.g., \cite{KaSaBu06,Gu97}, and the references therein. 
In the solutions of the problems dealing with gapped palindromes inside a word $w$ (Problems \ref{LPFgG}(a), \ref{LPFg(i)}(a), and \ref{LLAP}(a)) we construct the suffix array and $\LCP$ data structures for 
the word $u=w0w^R$, where $0$ is a new symbol lexicographically smaller than all the symbols of $V$; this takes $\bigo(|w|)$ time. To check whether $w[i..j]$ occurs at position $\ell$ in $w$ (respectively, $w[i..j]^R$ occurs at position $\ell$ in $w$) we check whether $\ell+(j-i+1)\leq n$ and $LCP_u(i,\ell)\geq j-i+1$ (respectively, $\LCP_u(\ell,2|w|-j+w))\geq j-i+1$). To keep the notation free of subscripts, when we measure the longest common prefix of a word $w[1..j]^R$ and word $w[i..n]$ we write $\LCP(w[1..j]^R,w[1..i])$, and this is in fact an $\LCP$-query on $u=w0w^R$; when we measure the longest common prefix of a word $w[j..n]$ and word $w[i..n]$ we write $\LCP(j,i)$, and this is in fact an $\LCP$-query on $w$.

The suffix array of $w0w^R$ allows us to construct in linear time a list ${\mathcal L}$ of the suffixes $w[i..n]$ of $w$ and of the mirror images $w[1..i]^R$ of the prefixes of $w$ (which correspond to the suffixes of length less than $|w|$ of $w0w^R$), ordered lexicographically. Generally, we denote by $Rank[i]$ the position of $w[i..n]$ in the ordered list ${\cal L}$ of these factors, and by $Rank_R[i]$ the position of $w[1..i]^R$ in ${\mathcal L}$. 

The dictionary of basic factors (introduced by~\cite{DBF}) of a word $w$ (DBF for short) is a data structure that labels the factors $w[i..i+2^k-1]$ (called basic factors), for $k\geq 0$ and $1\leq i\leq n-2^k+1$, such that every two identical factors of $w$ get the same label and we can retrieve~the~label of any basic factor in $\bigo(1)$ time. The DBF of a word of length $n$ is constructed in $\bigo(n\log n)$ time. 

Note that a basic factor $w[i..i+2^k-1]$ occurs either at most twice in any factor $w[j..j+2^{k+1}-1]$ or the occurrences of $w[i..i+2^k-1]$ in $w[j..j+2^{k+1}-1]$ form a run of period $\per(w[i..i+2^k-1])$ (so the corresponding positions where $w[i..i+2^k-1]$ occurs in $w[j..j+2^{k+1}-1]$ form an arithmetic progression of ratio $\per(w[i..i+2^k-1])$, see~\cite{KociumakaSPIRE2012}). Hence, the occurrences of  $w[i..i+2^k-1]$ in $w[j..j+2^{k+1}-1]$ can be presented in a compact manner: either at most two positions, or the starting position of the progression and its ratio. For $c\geq 2$, the occurrences of the basic factor $w[i..i+2^k-1]$ in $w[j..j+c 2^{k}-1]$ can be also presented in a compact manner: the positions (at most $c$) where $w[i..i+2^k-1]$ occurs isolated (not inside a run) and/or at most $c$ maximal runs that contain the overlapping occurrences of $w[i..i+2^k-1]$, each run having period $\per(w[i..i+2^k-1])$.

\begin{figure}\begin{center}
\includegraphics[width=\linewidth]{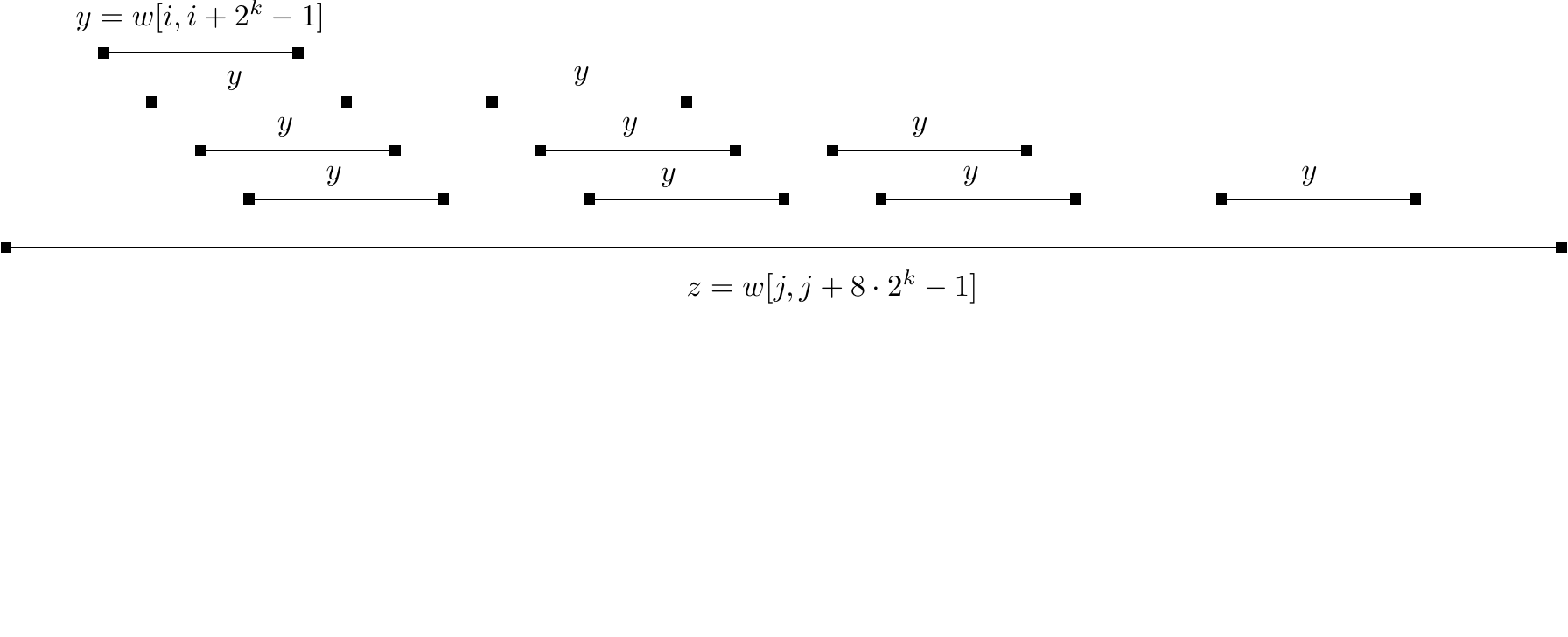}
\end{center}
\vspace*{-2.5cm}
\caption{Occurrences of the basic factors $y=w[i..i+2^k-1]$ in $z=w[j..j+8\cdot2^k-1]$. The overlapping occurrences are part of runs, and they can be returned as the pair formed of the first occurrence of $y$ from each run and the period of $y$. The representation of the occurrences of $y$ in $z$ will return $4$ elements: $3$ runs and one separate occurrence.}
\end{figure}

\begin{remark}\label{rem_DBF}
Using the DBF of a word $w$ of length $n$, given a number $\ell>0$ we can produce in $\bigo(n\log n)$ time a data structure answering the following type of queries in $\bigo(1)$ time: ``Given $i$ and $k$ return the compact representation of the occurrences of the basic factor $w[i..i+2^k-1]$ inside the basic factor $w[i-\ell-2^{k+1}..\ i-\ell-1]$''. Similarly, given $\ell>0$ and a constant $c>0$ (e.g., $c=10$), we can produce in $\bigo(n\log n)$ time a data structure answering the following type of queries in $\bigo(1)$ time: ``Given $i$ and $k$ return the compact representation of the occurrences of the basic factor $w[i..i+2^k-1]$ in $w[i-\ell-c2^{k}..i-\ell-1]$ ''. 

Indeed, once we construct the dictionary of basic factors of $w$, we reorganise it such that for each distinct basic factor we have an array with all the positions where it occurs, ordered increasingly. Now, we traverse each such array, keeping track of the current occurrence $w[i..i-2^k-1]$ and a window containing its occurrences from the range between $i-\ell-c 2^{k}$ and $i-\ell$ (and a compact representation of these occurrences); when we move in our traversal of this array to the next occurrence of the current basic factors, we also slide the window in the array, and, looking at the content of the previous window and keeping track of the occurrences that were taken out and those that were added to its content, we can easily obtain in constant time a representation of the occurrences of the considered basic factors inside the new window. 
\end{remark}

The previous remark is extended by the following lemma in a more general setting: $c$ is no longer a constant, and in a query we look, this time, for the occurrences of a basic factor $w[i..i+2^k-1]$ in factors $z=w[j..j+c 2^k -1]$, where there is no relation between $i$ and $j$. 
\begin{lemma}\label{find_occ_range}
Given a word $w$ of length $n$ and a number $c\geq 2$, we can preprocess $w$ in time $\bigo(n\log n)$ such that given any basic factor $y=w[i..i+2^k-1]$ and any factor $z=w[j..j+c 2^k -1]$, with $k\geq 0$, we can compute in $\bigo(\log \log n + c)$ time a compact representation of all the occurrences of $y$ in $z$.
\end{lemma}
\begin{proof}
We construct the dictionary of basic factors of a word of length $n$ in $\bigo(n\log n)$ time and reorganise it such that for each basic factor we have an array with all its occurrences, ordered by their starting positions. For each such array we construct data structures that allow predecessor/successor search in $\bigo(\log\log n)$ time with linear time preprocessing w.r.t. its length (see, e.g., \cite{Boas75}). When we have to return the occurrences of $y=w[i..i+2^k-1]$ in $z=w[j..j+c 2^k -1]$, we search in the basic factors-array corresponding to $w[i..i+2^k-1]$ the successor of $j$ and, respectively, the predecessor of $j+c2^k -1$ and then return a compact representation of the occurrences of $w[i..i+2^k-1]$ between these two values. This representation can be obtained in $\bigo(c)$ time. We just have to detect those occurrences that form a run; this can be done with a constant number of $\LCP$ queries. Indeed, for two consecutive occurrences, we compute the length of their overlap, which gives us a period of $w[i..i+2^k-1]$. Then we look in $w$ to see how long can the run with this period be extended to the right, which gives us the number of occurrences of  $w[i..i+2^k-1]$ in that run. As their starting positions form an arithmetic progression, we can represent them compactly. So, we return the representation of the occurrences of $w[i..i+2^k-1]$ from this run, and then move directly to the first occurrence of $w[i..i+2^k-1]$ appearing after this run and still in the desired range; as there are at most $\bigo(c)$ runs and separate occurrences of the given basic factor in the desired range, the conclusion follows.
\end{proof}

In the following we make use of a result by \cite{Gawrychowski11}, where it is shown that one can also construct data structures that allow the fast identification of the suffixes of a word that start with a given basic factor.
\begin{lemma}(\cite{Gawrychowski11})\label{find_node}
A word $w$ of length $n$ can be processed in $\bigo(n)$ time such that given any basic factor $w[i..i+2^k-1]$, with $k\geq 0$, we can retrieve in $\bigo(1)$ time the range of the suffix array of $w$ of suffixes starting with $w[i..i+2^k-1]$. Equivalently, we can find the node of the suffix tree of $w$ closest to the root, such that the label of the path from the root to that node has the prefix $w[i..i+2^k-1]$.
\end{lemma}
We can now show the following lemma. 
\begin{lemma}\label{find_occ_small}
Given a word $v$, $|v|=\alpha \log n$, we can process $v$ in time $\bigo(\alpha \log n)$ time such that given any basic factor $y=v[j\cdot2^k+1..(j+1)2^{k}]$, with $j,k\geq 0$ and $j2^k+1> (\alpha-1)\log n$, we can find in $\bigo(\alpha)$ time $\bigo(\alpha)$ bit-sets, each storing $\bigo(\log n)$ bits, characterising all the occurrences of $y$ in $v$.
\end{lemma}
\begin{proof}
We first show how the proof works for $\alpha=1$. 

We first build the suffix tree for $v$ in $\bigo(\alpha \log n)$ time (see~\cite{Farach97}). We further process this suffix tree such that we can find in constant time, for each factor $v[j\cdot 2^k+1..(j+1)2^{k}]$, the node of the suffix tree which is closest to the root with the property that the label of path from the root to it starts with $v[j\cdot 2^k+1..(j+1)2^{k}]$. According to Lemma \ref{find_node}, this can be done in linear time. 

Now, we augment the suffix tree in such a manner that for each node we store an additional bit-set, indicating the positions of $v$ where the word labelling the path from the root to the respective node occurs. Each of these bit-sets, of size $\bigo(\log n)$, can be stored in constant space; indeed, each $\log n$ block of bits from the bit-set can be seen in fact as a number between $1$ and $n$ so we only need to store a constant number of numbers smaller than $n$; in our model, each such number fits in a memory word. Computing the bit-sets can be done in a bottom up manner in linear time: for a node, we need to make the union of the bit-sets of its children, and this can be done by doing a bitwise {\bf or} operation between all the bit-sets corresponding to the children. So, now, checking the bit-set associated to the lowest node of the suffix tree such that the label of the path from the root to that node starts with $v[j\cdot 2^k+1..(j+1)2^{k}]$ we can immediately output a representation of this factor's occurrences in $v$.

This concludes the proof for $\alpha=1$. 

For $\alpha>1$, we just have to repeat the algorithm in the previous proof for the words $v[i\log n+1..(i+2)\log n]v[(\alpha-1)\log n+1 ..\alpha\log n]$, for $0\leq i\leq \alpha-2$, which allows us to find all the occurrences of the basic factors of $v[(\alpha-1)\log n+1 ..\alpha\log n]$ in $v$. The time is clearly $\bigo(\alpha \log n)$. 

\end{proof}

Note that each of the bit-sets produced in the above lemma can be stored in a constant number of memory words in our model of computation. Essentially, this lemma states that we can obtain in $\bigo(\alpha \log n)$ time a representation of size $\bigo(\alpha)$ of all the occurrences of $y$ in $v$. 

\begin{remark}\label{find_occ_small_range}
By Lemma \ref{find_node}, given a word $v$, $|v|=\alpha \log n$,  and a basic factor $y=v[j\cdot2^k+1,(j+1)2^{k}]$, with $j,k\geq 0$ and $j2^k+1> (\alpha-1)\log n$, we can produce $\bigo(\alpha)$ bit-sets, each containing exactly $\bigo(\log n)$ bits, characterising all the occurrences of $y$ in $v$. Let us also assume that we have access to all values $\log x$ with $x\leq n$ (which can be ensured by a $\bigo(n)$ preprocessing). Now, using the bit-sets encoding the occurrences of $y$ in $v$ and given a factor $z$ of $v$, $|z|=c|y|$ for some $c\geq 1$, we can obtain in $\bigo(c)$ time the occurrences of $y$ in $z$: the positions (at most $c$) where $y$ occurs outside a run and/or at most $c$ runs containing the occurrences of $y$. Indeed, the main idea is to select by bitwise operations on the bit-sets encoding the factors of $v$ that overlap $z$ the positions where $y$ occurs (so the positions with an $1$). For each two consecutive such occurrences of $y$ we detect whether they are part of a run in $v$ (by $\LCP$-queries on $v$) and then skip over all the occurrences of $y$ from that run (and the corresponding parts of the bit-sets) before looking again for the $1$-bits in the bit-sets.
 \end{remark}

Some of our solutions rely on an efficient solution for the \emph{disjoint set union-find} problem. This problem asks to maintain a family consisting initially of $d$ disjoint singleton sets from the universe $U=[1,n+1)$ (shorter for $\{1,\ldots,n\}$) so that given any element we can locate its current set and return the minimal (and/or the maximal) element of this set (operation called {\em find-query}) and we can merge two disjoint sets into one (operation called {\em union}). In our framework, we know from the beginning the pairs of elements whose corresponding sets can be joined. Under this assumption, a data-structure fulfilling the above requirements can be constructed in $\bigo(d)$ time such that performing a sequence of $m$ find and union operations takes $\bigo(m)$ time in the computation model we use (see \cite{Gabow83}). As a particular case, this data structure solves with $\bigo(n)$ preprocessing time and $\bigo(1)$ amortised time per operation the \emph{interval union-find} problem, which asks to maintain a partition of the universe $U=[1,n+1)$ into a number of $d$ disjoint intervals, so that given an element of $U$ we can locate its current interval, and we can merge two adjacent intervals of the partition. 

\begin{remark}\label{weighted_tree}As a first consequence of the algorithms of \cite{Gabow83}, we recall a folklore result, stating that we can process in $\bigo(n)$ time a weighted tree with $\bigo(n)$ nodes and all weights in $\bigo(n)$ so that we can compute off-line, also in linear time, the answer to $\bigo(n)$ weighted level ancestor queries on the nodes of this tree (where such a query asked for the first node on a path from a given node to the root such that the sum of the weights of the edges between the given and the returned node is greater than a given weight). 
\end{remark}

In the solutions of both tasks of Problem \ref{LPFgG}, we use the following variant of the interval union-find problem.
\begin{remark}\label{rem_Union_Find} 
Let $U=[1,n+1)$. We are given integers $d,\ell>0$ and for each~$i\leq \ell$ a partition ${\mathcal P}_i$ of $U$ in $d$ intervals, a sequence of $d$ find-queries and a sequence of $d$ union-operations to be performed alternatively on ${\mathcal P}_i$ (that is, we perform one find query then the one union operation, and so on, in the order specified in the sequences of union and find-queries). We can maintain the structure and obtain the answers to all find-queries in $\bigo(n+\ell d)$ total time. 

To see why this holds, assume that we are given the sequence $Q_i=\langle x^i_1,\ldots,x^i_d\rangle$ of $d$ find-queries and the sequence $U_i=\langle y^i_1,\ldots,y^i_d\rangle$ of $d$ union-operations to be performed alternatively on this partition. As defined above, first we perform the first find query in which we search for the interval containing $x_1$, then the first union between the interval ending on $y_1$ and the one starting on $y_1$, then the second find query on the updated partition, in which we search the interval containing $x_2$, and so on. In the following we show how to obtain the answers to all find-queries in $\bigo(n+\ell d)$ time. 

As a first step, we sort in time $\bigo(n+\ell d)$ all the ends of the intervals in ${\mathcal P}_i$ and the elements of $Q_i$ for all $i$ at once, using radix-sort. Then we separate them (now sorted) according to the partition they correspond to: so, we have for each $i$ an ordered list of the ends of the intervals of ${\mathcal P}_i$ and the elements of $Q_i$. Now we basically have all the data needed to be able to run the the algorithms of~\cite{Gabow83} in $\bigo(d)$ time for each partition (e.g., the time used in our setting by the sub-routines defined and used by \cite{Gabow83} is exactly just the one the required to apply the results of that paper, for each of the partitions). In conclusion, the total time needed to process one partition and to perform all the find queries and union operations on it is $\bigo(d)$. This adds up to a total time of $\bigo(n+\ell d)$. 
\end{remark}

We further give another lemma related to the union-find data structure.
\begin{lemma}\label{stabbing}
Let $U=[1,n+1)$. We are given $k$ intervals $I_1,\ldots, I_k$ included in~$U$. Also, for each $j\leq k$ we are given a positive integer $g_j\leq n$, the weight of the interval $I_j$. We can compute in $\bigo(n+k)$ time the values $H[i]=\max\{g_j\mid i\in I_j\}$ (or, alternatively, $h[i]=\min\{g_j\mid i\in I_j\}$) for all $i\leq n$.  
\end{lemma}
\begin{proof}
Let us assume that $I_j=[a_j,b_j)$ where $1\leq a_j\leq n$ and $1\leq b_j\leq n+1$, for all $1\leq j\leq k$. We first show how the array $H[\cdot]$ is computed.  

We sort the intervals $I_1, I_2,\ldots, I_k$ with respect to their starting positions $a_j$, for $1\leq j\leq k$. Then, we produce for each $g$ from $1$ to $n$ the list of the intervals $I_j$ that have $g_j=g$ (again, sorted by their starting positions). Using radix-sort we can achieve this in time $\bigo(n+k)$. Further, we set up a disjoint set union-find data structure for the universe $U=[1,n]$. Initially, the sets in our structure are the singletons $\{1\}.\ldots,\{n\}$; the only unions that we can a make while maintaining this structure are between the set containing $i$ and the set containing $i+1$, for all $1\leq i\leq n-1$. Therefore, we can think that our structure only contains intervals $[a,b)$; initially, we have in our union-find data structures, for all $i\leq n$, the sets $\{i\}=[i,i+1)$. Now, we process the intervals that have weight $g$, for each $g$ from $n$ to $1$ in decreasing order; that is, the intervals are considered in decreasing order of their weight. So, assume that we process the input intervals of weight $g$. Assume that at this moment, $U$ is partitioned in some intervals (some of them just singletons); initially, as mentioned above, we only have singletons $[i,i+1)$. Let $I_j=[a,b)$ be an input interval of weight $g$. Let now $\ell=a$ (in the following, $\ell$ will take different values between $a$ and $b$). We locate the interval $[c_1,c_2)$  (of the union-find structure we maintain) where $\ell $ is located; if this is a singleton and $H[\ell ]$ was not already set, then we set $H[\ell ]=g$. Further, unless $\ell $ is $a$, we merge the interval containing $\ell $ to that containing $\ell-1$ and set $\ell =c_2$; if $\ell =a$ we just set $\ell =c_2$.  In both cases, we repeat the procedure while $\ell <b$. We process in this manner all the intervals $I_j$ with $g_j=g$, in the order given by their starting positions, then continue and process the intervals with weight $g-1$, and so~on.

The algorithm computes $H[\cdot]$ correctly. Indeed, we set $H[\ell ]=g$ when we process an interval $I_j$ of weight $g$ that contains $\ell $, and no other interval of greater weight contained $\ell $. To see the complexity of the algorithm we need to count the number of union and find operations we make. First, we count the number of union operations. For this, it is enough to note that for each element $\ell $ we might make at most $2$ unions: one that unites the singleton interval $[\ell ,\ell +1)$ to the interval of $\ell -1$, which has the form $[a,\ell )$ for some $a$, and another one that unites the interval of $\ell +1$ to the one of $\ell $. So, this means that we make at most $\bigo(n)$ union operations. For the find operations, we just have to note that when an interval $I_j$ is processed the total number of finds is $\bigo(|I_j \setminus U|+2)$, where $U$ is the union of the intervals that were processed before $I_j$. This shows that the total number of find operations is $\bigo(k+|\cup_{j=1,k}I_j|)=\bigo(n+k)$. 

By the results of \cite{Gabow83}, as we know the structure of the unions we can make forms a tree, our algorithm runs in $\bigo(n+k)$ time. 

The computation of $h[\cdot]$ is similar. The only difference is that we consider the intervals in increasing order of their weight.
\end{proof}

We conclude this section with a lemma that will be used in the solutions of Problem \ref{LLAP}. It shows how to compute in linear time the length of the longest square centred at each position of a given word. 
\begin{lemma}\label{centred_squares}
Given a word $w$ of length $n$ we can compute in $\bigo(n)$ time the values $SC[i]=\max\{|u|\mid u$ is both a suffix of $w[1..i-1]$ and a prefix of $w[i..n]\}$.
\end{lemma}
\begin{proof}
Note that each square $u^2$ occurring in a word $w$ is part of a maximal run $w[i'..j']=p^\alpha p'$, where $p$ is primitive and $p'$ is a prefix of $p$, and $u=q^\ell$, where $q$ is a cyclic shift of $p$ (i.e., $q$ is a factor of $p^2$ of length $|p|$) and $\ell \leq \frac{\alpha}{2}$. 

So, if we consider a maximal run $r=p^\alpha p'$ and some $\ell\leq \frac{\alpha}{2}$, we can easily detect the possible centre positions of the squares having the form $(q^\ell)^2$ contained in this run, with $q$ a cyclic shift of $p$. These positions occur consecutively in the word $w$: the first is the $(|p|\ell+1)^{th}$ position of the run, and the last is the one where the suffix of length $|p|{\ell}$ of the run starts. So they form an interval $I_{r,\ell}$ and we associate to this interval the weight $g_{r,\ell}=|p|\ell$ (i.e., the length of an arm of the square). In this way, we define $\bigo(n)$ intervals (as their number is upper bounded by the sum of the exponents of the maximal runs of $w$), all contained in $[1,n+1)$, and each interval having a weight between $1$ and $n$. By Lemma \ref{stabbing}, we can process these intervals so that we can determine for each $i\in [1,n+1)$ the interval of maximum weight containing $i$, or, in other words, the maximum length $SC[i]$ of a square centred on $i$. This procedure runs in $\bigo(n)$ time. 
 \end{proof}

It is worth noting that using the same strategy as in the proof of Lemma~\ref{centred_squares}, one can detect the minimum length of a square centred at each position of a word (also called the local period at that position) in linear time. Indeed, in this case we note that the square of minimum length centred at some position of the input word must be primitively rooted. Then, we repeat the same strategy as in the proof of Lemma \ref{centred_squares}, but only construct the intervals $I_{r,\ell}$ for the case when $\ell=1$ (i.e., the intervals $I_{r,1}$) for all runs $r$, with the corresponding weights. Then we use Lemma \ref{stabbing} for these intervals to determine the interval of minimum weight that contains each position of the word.
This leads to an alternative solution to the problem of computing the local periods of a word, solved by~\cite{MFCS_Lecroq}. Compared to the solution of \cite{MFCS_Lecroq}, ours uses a relatively involved data structures machinery (disjoint sets union-find structures, as well as an algorithm finding all runs in a word), but is much shorter and seems conceptually simpler as it does not require a very long and detailed combinatorial analysis of the properties of the input word. 

The same strategy allows solving the problem of computing in linear time, for integer alphabets, the length of the shortest (or longest) square ending (or starting) at each position of a given word; this improves the results by  \cite{kosarajuCPM,XuCPM}, where such a result was only shown for constant size alphabets. Let us briefly discuss the case of finding the shortest square ending at each position of the word; such squares are, clearly,  primitively rooted. So, just like in the case of searching squares centred at some position, we can determine for each run an interval of positions where a square having exactly the same period as the run may end. The period of each run becomes the weight of the interval associated to that run. Then we use again Lemma \ref{stabbing}  for these intervals to determine the interval of minimum weight that contains each position of the word. Thus, we determine for each position the shortest square ending on it. In the case when we want the longest square ending at each position, just like in the case of centred squares, we define for each run several intervals. Fortunately, this does not increase asymptotically the complexity.

In the following we apply the machinery developed in this section to identify the longest gapped repeats and palindromes occurring in a word.

\section{Lower and upper bounded gap}
In this section we present the solutions of Problems \ref{LPFgG}(a) and \ref{LPFgG}(b). With these problems solved, we can immediately retrieve the longest gapped repeat and palindrome contained in the input word of these problems, such that the length gap is between a given lower bound and a given upper bound.

\begin{theorem}\label{thLPrFgG}
Problem \ref{LPFgG}(a) can be solved in linear time.
\end{theorem}
\begin{proof}
Let $\delta=G-g$; for simplicity, assume that $n$ is divisible by $\delta$. Further, for $1\leq i\leq n$, 
let $u$ be the longest factor which is a prefix of $w[i..n]$ such that $u^R$ is a suffix of some $w[1..k]$ with $g< i-k \leq G$; then, $B[i]$ is the rightmost position $j$ such that $g< i-j\leq G$ and $u^R$ is a suffix of $w[1..j]$. Knowing $B[i]$ means knowing $\LPF_{g,G}[i]$: we just have to return $\LPF(w[i..n],w[1..j]^R)$. 

We split the set $\{1,\ldots,n\}$ into $\frac{n}{\delta}$ ranges of consecutive numbers: $I_0=\{1,2,\ldots,\delta\}$, $I_1=\{\delta+1, \delta+2,\ldots, 2\delta\}$, and so on. Note that for all $i$ in some range $I_k=\{k\delta+1,k\delta+2,\ldots,(k+1)\delta\}$ from those defined above, there are at most three consecutive ranges where some $j$ such that $g< i-j\leq G$ may be found: the range containing $k\delta-G+1$, the range containing $(k+1)\delta-g-1$, and the one in between these two (i.e., the range containing $k\delta+1-g$). Moreover, for some fixed $i$, we know that we have to look for $B[i]$ in the interval $\{i-G,\ldots,i-g-1\}$, of length $\delta$; when we search $B[i+1]$ we look at the interval $\{i+1-G,\ldots,i-g\}$. So, basically, when trying to find $B[i]$ for all $i\in I_k$, we move a window of length $\delta$ over the three ranges named above, and try to find for every content of the window (so for every $i$) its one element that fits the description of $B[i]$. The difference between the content of the window in two consecutive steps is not major: we just removed an element and introduced a new one. Also, note that at each moment the window intersects exactly two of the aforementioned three ranges. We try to use these remarks, and maintain the contents of the window such that the update can be done efficiently, and the values of $B[i]$ (and, implicitly, $\LPF_{g,G}[i]$) can be, for each $i\in I$, retrieved very fast. Intuitively, grouping the $i$'s on ranges of consecutive numbers allows us to find the possible places of the corresponding $B[i]$'s for all $i\in I_k$ in $\bigo(\delta)$ time.

\begin{figure}[H]
\begin{center}
\includegraphics[width=\linewidth]{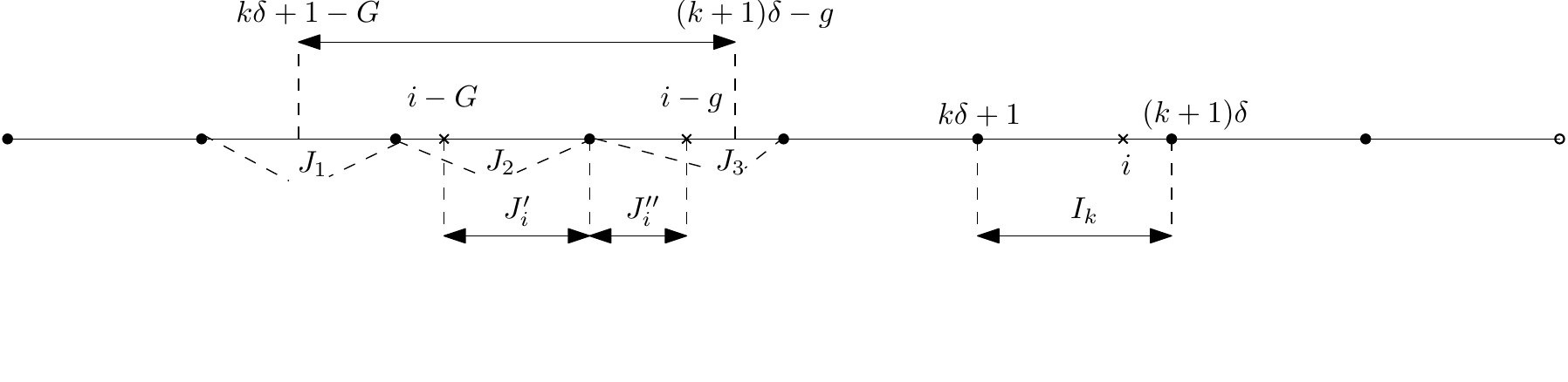}
\end{center}
\vspace{-1cm}
\caption{Proof of Theorem \ref{thLPrFgG}: Construction of the ranges $J_1,J_2,J_3$ for some range $I_k$, and of the sliding window $J'_i\cup J''_i$ for some $i\in I_k$}
\end{figure}

We now go into more details. As we described in the preliminaries, let ${\mathcal L}$ be the lexicographically ordered list of the the suffixes $w[i..n]$ of $w$ and of the mirror images $w[1..i]^R$ of the prefixes of $w$ (which correspond to the suffixes of $w^R$). For the list ${\mathcal L}$ we compute the arrays $Rank[\cdot]$ and $Rank_R[\cdot]$. 
We use the suffix array for $w0w^R$ to produce for each of the ranges $I_k$ computed above the set of suffixes of $w$ that start in the respective range (sorted lexicographically, in the order they appear in the suffix array) and the set of prefixes of $w$ ending in $I$, ordered lexicographically with respect to their mirror image. 

We consider now one of the ranges of indexes $I_k=\{k\delta+1,k\delta+2,\ldots, (k+1)\delta\}$ from the above, and show how we can compute the values $B[i]$, for all $i\in I_k$. For some $i\in I_k$ we look for the maximum $\ell$ such that there exists a position $j$ with $w[j-\ell+1..j]^R=w[i..i+\ell-1]$, and $i-G\leq j\leq i-g$. As already explained, for the $i$'s of $I_k$ there are three consecutive ranges where the $j$'s corresponding to the $i$'s of $I_k$ may be found; we denote them $J_1,J_2,J_3$. 

Now, for an $i\in I_k$ we have that $B[i]$ (for which $g<i-B[i]\leq G$) belongs to $J'_i\cup J''_i$, where $J'_i$ is an interval starting with $i-G$ and extending until the end of the range that contains $i-G$ (which is one of the ranges $J_1,J_2,J_3$) and $J''_i$ is an interval ending with $i-g-1$, which starts at the beginning of the range that contains $i-g-1$ (which is the range that comes next after the one containing $i-G$). Referencing back to the intuitive explanation we gave at the beginning of this proof, $J'_i\cup J''_i$ is the window we use to locate the value of $B[i]$ for an $i\in I_k$.  

To compute $B[i]$ for some $i$, take $f_i\in J'_i$ such that $\LCP(w[1..f_i]^R,w[i..n])\geq \LCP(w[1..j']^R,w[i..n])$ for all $j'\in J'_i$. Similarly, take $s_i\in J''_i$ such that for all $j'\in J''_i$ we have $\LCP(w[1..s_i]^R,w[i..n])\geq \LCP(w[1..j']^R,w[i..n])$. Once $s_i$ and $f_i$ computed, we set $B[i]=s_i$ if $\LCP(w[1..s_i]^R,w[i..n])\geq \LCP(w[1..f_i]^R,w[i..n])$; we set $B[i]=f_i$, otherwise. So, in order to compute, for some $i$, the value $B[i]$, that determines $\LPF_{g,G}[i] $,  we first compute $f_i$ and $s_i$.

We compute for all the indices $i\in I_k$, considered in increasing order, the values $f_i$. 
We consider for each $i\in I_k$ the interval $J'_i$ and note that $J'_i\setminus J'_{i+1}=\{i-G\}$, and, if $J'_i$ is not a singleton (i.e., $J'_i\neq \{i-G\}$) then $J'_{i+1}\subset J'_i$. 
If $J'_i$ is a singleton, than $J'_{i+1}$ is, in fact, one of the precomputed range $I_{p}$, namely the one which starts on position $i+1-G$ (so, $p=\frac{i-G}{\delta}$). 
 
These relations suggest the following approach. We start with $i=k\delta+1$ and consider the set of words $w[1..j]^R$ with $j\in J'_i$; this set can be easily obtained in $\bigo(\delta)$ time by finding first the range $J_1$ in which $i-G$ is contained (which takes $\bigo(1)$ time, as $J_1=I_p$ for $p=\left\lfloor\frac{i-G}{\delta}\right\rfloor$), and then selecting from $J_1$ of the set of prefixes $w[1..d]$ of $w$, ending in $J_1$ with $d\geq i-G$ (ordered lexicographically with respect to their mirror image). The ranks corresponding to these prefixes in the ordered list ${\mathcal L}$ (i.e., the set of numbers $Rank_R[d]$) define a partition of the universe $U=[0,2n+1)$ in at most $\delta+2$ disjoint intervals. So, we can maintain an interval union-find data structures like in Remark \ref{rem_Union_Find}, where the ranks are seen as limits of the intervals in this structure. We assume that the intervals in our data structure are of the form $[a,b)$, with $a$ and $b$ equal to some $Rank_R[d_a]$ and $Rank_R[d_b]$, respectively. The first interval in the structure is of the form $[0,a)$, while the last is of the form $[b,2n+1)$. We now find the interval to which $Rank[i]$ belongs; say that this is $[a,b)$. This means that the words $w[1..d_a]^R$ and $w[1..d_b]^R$ are the two words of $\{w[1..d]^R\mid d\in J'_i\}$ which are closest to $w[i..n]$ lexicographically ($w[1..d_a]^R$ is lexicographically smaller, $w[1..d_b]$ is greater). Clearly, $f_i=d_a$ if $\LCP(w[1..d_a]^R,w[i..n])\geq \LCP(w[1..d_b]^R,w[i..n])$ and $f_i=d_b$, otherwise (in case of a tie, we take $f_i$ to be the greater of $d_a$ and $d_b$). So, to compute $f_i$ we query once the union-find data structure to find $a$ and $b$, and the corresponding $d_a$ and $d_b$, and then run two more $\LCP$ queries. 

When moving on to compute $f_{i+1}$, we just have to update our structure and then run the same procedure. Now, $i-G$ is no longer a valid candidate for $f_{i+1}$, and it is removed from $J'_i$. So we just delete it from the interval union-find data structure, and merge the interval ending right before $Rank_R[i-G]$ and the one starting with $Rank_R[i-G]$. This means one union operation in our interval union-find structure. Then we proceed to compute $f_{i+1}$ as in the case of $f_i$.

The process continues until $J'_i$ is a singleton, so $f_i$ equals its single element.

Now, $i-G$ is the last element of one of the ranges $J_1,J_2,$ or $J_3$; assume this range is $I_p$. So far, we performed alternatively at most $\delta$ find queries and $\delta$ union operations on the union-find structure. Now,  instead of updating this structure, we consider a new interval partition of $[0,2n+1)$ induced by the ranks of the~$\delta$ prefixes ending in $I_{p+1}$. When computing the values $f_i$ for $i\in I_k$ we need to consider a new partition of $U$ at most once: at the border between $J_1$ and $J_2$.

It is not hard to see from the comments made in the above presentation that our algorithm computes $f_i\in I_k$ correctly. In summary, in order to compute $f_i$, we considered the elements $i\in I_k$ from left to right, keeping track of the left part $J'_i$ of the window $J'_i\cup J''_i$ while it moved from left to right through the ranges $J_1,J_2$ and $J_3$. Now, to compute the values $s_i$ for $i\in I_k$, we proceed in a symmetric manner: we consider the values $i$ in decreasing order, moving the window from right to left, and keep track of its right part $J''_i$. 

As already explained, by knowing the values $f_i$ and $s_i$ for all $i\in I_k$ and for all $k$, we immediately get $B[i]$ (and, consequently $\LPF_{g,G}[i]$) for all $i$. 

We now evaluate the running time of our approach. We can compute, in $\bigo(n)$ time, from the very beginning of our algorithm the partitions of $[1,2n-1)$ we need to process (basically, for each $I_k$ we find $J_1$, $J_2$ and $J_3$ in constant time, and we get the three initial partitions we have to process in $\bigo(\delta)$ time), and we also know the union-operations and find-queries that we will need to perform for each such partition (as we know the order in which the prefixes are taken out of the window, so the order in which the intervals are merged). In total we have $\bigo(n/\delta)$ partitions, each having initially $\delta+2$ intervals, and on each we perform $\delta$ find-queries and $\delta$-union operations. So, by Remark \ref{rem_Union_Find}, we can preprocess this data (once, at the beginning of the algorithm) in $\bigo(n)$ time, to be sure that the time needed to obtain the correct answers to all the find queries is $\bigo(n)$. So, the total time needed to compute the values $f_i$ for all $i\in I_k$ and for all $k$ is $\bigo(n)$. Similarly, the total time needed to compute the values $s_i$ for all $i$ is $\bigo(n)$. Then, for each $i$ we get $B[i]$ and $\LPF_{g,G}[i]$ in $\bigo(1)$ time.

Therefore, Problem \ref{LPFgG}(a) can be solved in linear time.
\end{proof}

Before giving the formal solution for Problem \ref{LPFgG}(b), we give a short intuitive description of how this algorithm is different from the previous one. In this case, when trying to construct the repeat $uvu$ with the longest arm occurring at a certain position, we need to somehow restrict the range where its left arm may occur. To this end, we restrict the length of the arm, and, consequently, search for $u$ with $2^k\leq |u|\leq 2^{k+1}$, for each $k\leq \log n$. Such a factor always start with $w[i+1..i+2^k]$ and may only occur in the factor $w[i-G-2^{k+1}..i-g]$. If $2^k\geq G-g$, using the dictionary of basic factors data structures (from~\cite{DBF}) and the results by \cite{KociumakaSPIRE2012}, we get a constant size representation of the occurrences of $w[i+1..i+2^k]$ in that range (which are a constant number of times longer than the searched factor), and then we detect which one of these occurrences produces the repeat $uvu$ with the longest arm. If $2^k<G-g$, we use roughly the same strategy to find the repeat $uvu$ with the longest arm and $u$ starting in a range of radius $2^{k+1}$ centred around $i-G$ or in a range of length $2^{k+1}$ ending on $i-g$ (again, these ranges are just a constant number of times longer than the searched factor). To detect a repeat starting between $i-G+2^{k+1}$ and $i-g-2^{k+1}$ we use the strategy from the solution of Problem \ref{LPFgG}(a); in that case, we just have to return the longest common prefix of $w[i..n]$ and the words $w[j..n]$ with $i-G+2^{k+1}\leq j\leq i-g-2^{k+1}$. Overall, this approach can be implemented to work in $\bigo(n \log n)$ time.
\begin{theorem}\label{sol_LPFgG}
Problem \ref{LPFgG}(b) can be solved in $\bigo(n\log n)$ time.
\end{theorem}
\begin{proof}
For $1\leq i\leq n$, let $B[i]$ denote the value $j$ such that $w[j..j+\LPdF_{g,G}[i]-1]=w[i..i+\LPdF_{g,G}[i]-1]$ and $g<i-(j+\LPdF_{g,G}[i]-1)\leq G$. 
In other words, to define $B[i]$, let $u$ be the longest factor which is both a prefix of $w[i..n]$ and a suffix of some $w[1..k]$ with $g< i-k \leq G$; then, $B[i]$ is the rightmost position $j$ such that $g< i-(j+|u|-1)\leq G$ and $u$ is a suffix of $w[1..j+|u|-1]$. Clearly, knowing $B[i]$ means knowing $\LPF_{g,G}[i]$. 

Intuitively, computing $B[i]$ is harder in this case than it was in the case of the solution of Problem \ref{LPFgG}(a). Now we have no real information where $B[i]$ might be, we just know the range of $w$ where the longest factor that occurs both at $B[i]$ and at $i$ ends. So, to determine $B[i]$ we try different variants for the length of this factor, and see which one leads to the right answer. 

Let $\delta=G-g$ and $k_0=\lfloor \log \delta\rfloor$; assume w.l.o.g. that $n$ is divisible by $\delta$. 

As already noted, the solution used in the case of gapped palindromes in the proof of Theorem \ref{thLPrFgG} does not work anymore in this case: we do not know for a given $i$ a range in which $B[i]$ is found. So, we try to restrict the places where $B[i]$ can occur. We split our discussion in two cases.

In the first case, we try to find, for each $i$ and each $k\geq k_0$, the factor $uvu$ with the longest $u$, such that the second $u$ occurs at position $i$ in $w$, $ 2^k \leq |u|< 2^{k+1}$ and $g<|v|\leq G$. Clearly, $u$ should start with the basic factor $w[i..i+2^k-1]$, and the left arm $u$ in the factor $uvu$ we look for should have its prefix of length $2^k$ (so, a factor equal to the basic factor $w[i..i+2^k-1]$) contained in the factor $w[i-G-2^{k+1}..i-g-1]$, whose length is \centerline{$G-g+2^{k+1}\leq 2^{k_0+1}+2^{k+1}-1\leq 2^{k+2}-1\leq 4\cdot 2^k.$} So, by Remark \ref{rem_DBF}, using the dictionary of basic factors we can 
retrieve a compact representation of the occurrences of $w[i..i+2^k-1]$ in $w[i-G-2^{k+1}..i-g-1]$: these consist in a constant number of isolated occurrences and a constant number of runs. For each of the isolated occurrences $w[j'..j'+2^k-1]$ we compute $\LCP(j',i)$ and this gives us a possible candidate for $u$; we measure the gap between the two occurrences of $u$ (the one at $j'$ and the one at $i$) and if it is between smaller than $G$, we store this $u$ as a possible solution to our problem (we might have to cut a suffix of $u$ so that the gap is also longer than $g$). Further, each of the runs has the form $p^\alpha p'$, where $|p|$ is the period of $w[i..i+2^k-1]$ and $p'$ is a prefix of $p$ (which varies from run to run); such a run cannot be extended to the right: either the period breaks, or it would go out of the factor $w[i-G-2^{k+1}..i-g-1]$. Using a $\LCP$ query we find the longest factor $p^\beta p''$ with period $|p|$, occurring at position $i$. For a run $p^\alpha p'$, the longest candidate for the first $u$ of the factor $uvu$ we look for starts with $p^\gamma p'''$ where $\gamma=\min\{\alpha,\beta\}$ and $p'''$ is the shortest of $p'$ and $p''$; if $p'=p''$, then this factor is extended with the longest factor that occurs both after the current run and after $p^\beta p''$ and does not overlaps the minimal gap (i.e., end at least $g$ symbols before $i$). This gives us 
a candidate for the factor $u$ we look for (provided that the gap between the two candidates for $u$ we found is not too large). In the end, we just take the longest of the candidates we identified (in case of ties, we take the one that starts on the rightmost position), and this gives the factor $uvu$ with the longest $u$, such that the second $u$ occurs at position $i$ in $w$, $ 2^k \leq |u|< 2^{k+1}$ and $g<|v|\leq G$. 

\begin{figure}\begin{center}
\includegraphics[width=\linewidth]{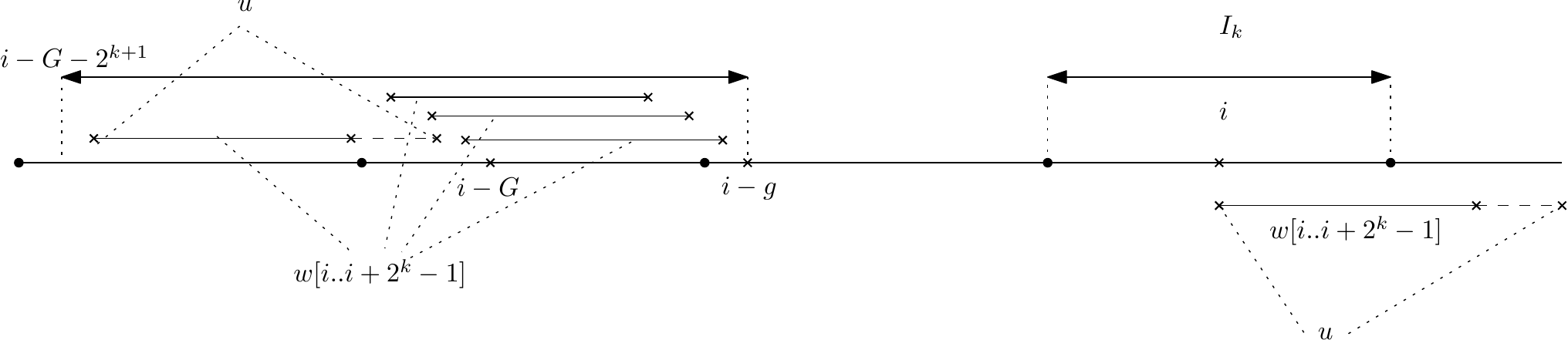}
\end{center}
\vspace{-0.5cm}
\caption{Proof of Theorem \ref{sol_LPFgG}: Occurrences of $w[i..i+2^k-1]$ inside $w[i-G-2^{k+1}..i-g-1]$: one separate occurrence, and a run containing three occurrences. The separate occurrence can be prolonged to produce a gapped palindrome $uvu$, which is not, however, long enough as it does not reach the range between $i-G$ and $i-g$. The rightmost occurrence in the run produces the gapped repeat $w[i..i+2^k-1] v' w[i..i+2^k-1]$, which fulfils the conditions imposed on the gap.}
\end{figure}

Iterating this process for all $k$ and $i$ as above, we obtain for each $i$ the factor $uvu$ with the longest $u$, such that the second $u$ occurs at position $i$ in $w$, $ 2^{k_0} \leq |u|$ and $g<|v|\leq G$.
By Remark \ref{rem_DBF}, the time spent in this computation for some $i$ and $k$ is constant, so the overall time used in the above computation is $\bigo(n\log n)$. Clearly, if for some $i$ we found such a factor $u$, then this gives us both $B[i]$ and $\LPdF_{g,G}[i]$; if not, we continue as follows.

In the second case, we try to identify  for each $i$ the factor $uvu$ with the longest $u$, such that the second $u$ factor occurs at position $i$ in $w$, $u<  2^{k_0}$ and $g<|v|\leq G$. Again, we consider each $k\leq k_0-1$ separately and split the discussion in three cases.

Firstly, for each $i$, we find the factor $uvu$ with the longest $u$, such that the second $u$ occurs at position $i$ in $w$, $2^k \leq u<  2^{k+1}$, $g<|v|\leq G$, and the first $u$ has its prefix of length $2^k$ in the factor $w[i-G-2^{k+1}..i-G+2^k]$, whose length is $2^k+2^{k+1}+1\leq 4\cdot 2^k$. This can be done similarly to the above (report all the occurrences of $w[i..i+2^k-1]$ in that range, and try to extend them to get $u$); for all $i$ and all $k$ takes $\bigo(n\log \delta)$ time. 

Secondly, for each $i$, we find the factor $uvu$ with the longest $u$, such that the second $u$ occurs at position $i$ in $w$, $2^k \leq u<  2^{k+1}$, $g<|v|\leq G$, and the first $u$ has its prefix of length $2^k$ in the factor $w[i-g-2^{k+1}..i-g-1]$, whose length is $2^{k+1}$. Again, this can be done like above, and for all $i$ and all $k$ takes $\bigo(n\log \delta)$ time. 

The third and more complicated subcase is when the first $u$ starts in the factor $w[i-G..i-g-2^{k+1}-1]$, of length $\delta_k=\delta-2^{k+1}$. Let $g_k=g+2^{k+1}+1$; obviously, $\delta_k=G-g_k$. Note that, in this case, every factor of length at most $2^{k+1}$ starting in $w[i-G..i-g-2^{k+1}-1]$ ends before $i-g$, so it is a valid candidate for the $u$ we look for.
In this case we can follow the algorithm from the proof of Theorem \ref{thLPrFgG}. We split the set $\{1,\ldots,n\}$ into ranges of consecutive numbers: $I_0=\{1,2,\ldots,\delta_k\}$, $I_1=\{\delta_k+1, \delta_k+2,\ldots, 2\delta_k\}$, and so on. For some $ I_\ell=\{\ell\delta_k+1,\ldots,(\ell+1)\delta_k\}$, considering all $i\in I_\ell$ there are three consecutive ranges from the ones defined above, where the first $u$ of the factor $uvu$ we look for may occur. The first (leftmost, with respect to its starting position) such range is the one containing $\ell\delta_k-G$, and let us denote it $J_1$; the last one (rightmost) is the one that contains $(\ell+1)\delta_k-g_k-1$, and we denote it by $J_3$. Clearly, between the ranges containing $\ell\delta_k-G$ and $(\ell+1)\delta_k-g_k-1$, respectively, there is exactly one complete range, call it $J_2$. 

Moreover, for a precise $i\in I_\ell$ the possible starting positions of the left arm $u$ of the repeat $uvu$ of the type we are searching for (i.e., with the gap between $g$ and $G$), with the second $u$ starting on $i$, form a contiguous range $J'_i\cup J''_i$ where $J'_i$ is an interval starting with $i-G$ and extending until the end of the range $J_p$ that contains $i-G$, and $J''_i$ is an interval ending with $i-g_k-1$, contained in $J_{p+1}$ (when $p<3$). Like before, $J'_i\cup J''_i$ can be seen as the content of a window that slides through $J_1,J_2,J_3$ while searching for $B[i]$. We denote by $f_i$ the position of $J'_i$ such that $LCP(f_i,i)$ is maximum among all such positions; we denote by $s_i$ the position of $J''_i$ such that $LCP(s_i,i)$ is maximum among all such positions. Then we just have to check where, at $f_i$ or $s_i$, occurs a longer factor that also occurs at $i$. We just explain how to compute $f_i$. 

We start with $i=\ell\delta_k+1$ and consider the set of words $w[j..n]$ with $j\in J'_i$; this set can be easily obtained in $\bigo(\delta)$ time by finding first the range $J_1$ in which $i-G$ is contained and then selecting from $J_1$ of the set of words $w[d..n]$ of $w$ starting in $J_1$ with $d\geq i-G$ (ordered lexicographically). The ranks corresponding to these suffixes in the suffix array of $w$ (i.e., the set of numbers $Rank[d]$) define a partition of the universe $U=[0,n+1)$ in at most $\delta_k+2$ disjoint intervals. So, we can maintain an interval union-find data structures like in Remark \ref{rem_Union_Find}, where the ranks are seen as limits of the intervals in this structure. We assume that the intervals in our data structure are of the form $[a,b)$, with $a$ and $b$ integers that are equal to some $Rank[d_a]$ and $Rank[d_b]$. The first interval in the structure is of the form $[0,a)$, while the last is of the form $[b,n+1)$. Recall that we want to compute $f_i$. To this end, we just have to find the interval to which $Rank[i]$ belongs; say that this is $[a,b)$. This means that the words $w[d_a..n]$ and $w[d_b..n]$ are the two words of the set $\{w[d..n]\mid d\in J'_i\}$ which are closest to $w[i..n]$ lexicographically ($w[d_a..n]$ is lexicographically smaller, while $w[d_b..n]$ is lexicographically greater). Clearly, $f_i=d_a$ if $\LCP(d_a,i)\geq \LCP(d_b,i)$ and $f_i=d_b$, otherwise (in case of a tie, we take $d_a$ if $d_b<d_a$, and $d_b$ otherwise). So, to compute $f_i$ we have to query once the union-find data structure to find $a$ and $b$, and the corresponding $d_a$ and $d_b$, and then compute the answer to two $\LCP$ queries. 

When moving on to compute $f_{i+1}$, we just have to update our structure. In this case, $i-G$ is no longer a valid candidate for $f_{i+1}$, as it should be removed from $J'_i$. So we just delete it from the interval union-find data structure, and merge the interval ending right before $Rank[i-G]$ and the one starting with $Rank[i-G]$. This means one union operation in our interval union-find structure. Then we proceed to compute $f_{i+1}$ just in the same manner as in the case of $f_i$.

The process continues in this way until we try to compute $f_i$ for $J'_i$ being a singleton. This means that $i-G$ is the last element of one of the ranges $J_1$, $J_2$, or $J_3$; let us assume that this range is $I_p$ from the ranges defined above. Clearly, this time $f_i$ is the single element of $J'_i$. Until now, we performed alternatively at most $\delta_k$ find operations and $\delta_k$ union operations on the union-find data structure. Further, instead of updating the union-find data structure, we consider a new interval partition of $[0,n+1)$ induced by the ranks of the $\delta_k$ suffixes starting in $I_{p+1}$. Note that when computing the values $f_i$ for $i\in I_\ell$ we need to consider a new partition of $U$ once: at the border between $J_1$ and $J_2$.

Now, by the same arguments as in the proof for gapped palindromes, the process takes $\bigo(n)$ in total for all intervals $I_\ell$ (so, when we iterate $\ell$), for each $k$. In this way we find the factor $uvu$ with the longest $u$, such that the second $u$ occurs at position $i$ in $w$, $2^k \leq u<  2^{k+1}$, $g<|v|\leq G$, and the first $u$ starts in $w[i-G..i-g-2^{k+1}-1]$. Iterating for all $k$, we complete the computation of $B[i]$ and $\LPdF_{g,G}[i]$ for all $i$; the needed time is $\bigo(n\log \delta)$ in this case.

Considering the three cases above leads to finding for each $i$ the factor $uvu$ with the longest $u$, such that the second $u$ factor occurs at position $i$ in $w$, $u<  2^{k_0}$ and $g<|v|\leq G$. The complexity of this analysis is $\bigo(n\log \delta)$. After concluding this, we get the value $B[i]$ for each $i$. Therefore, the entire process of computing $B[i]$ and $\LPdF_{g,G}[i]$ for all $i$ takes $\bigo(n\log n)$ time.
\end{proof}

We conclude this section with the following consequence of Theorems~\ref{thLPrFgG} and~\ref{sol_LPFgG}.

\begin{theorem}
Given $w$ of length $n$ and two integers $g$ and $G$, such that $0\leq g< G\leq n$, we can find in linear time the gapped palindrome $u^Rvu$ occurring in that word with the longest arm $u$ and $g\leq |v|<G$.\\
Given $w$ of length $n$ and two integers $g$ and $G$, such that $0\leq g< G\leq n$, we can find in $\bigo(n \log n)$ time the gapped repeat $uvu$ occurring in that word with the longest arm $u$ and $g\leq |v|<G$.
\end{theorem}
\begin{proof}
It is immediate that, after solving Problem \ref{LPFgG} for the word $w$, we just have to check which is the longest gapped palindrome (respectively, repeat) stored in the array $LPrF_{g,G}[\cdot]$ (respectively, $\LPdF_{g,G}[\cdot]$). 
 \end{proof}

\section{Lower bounded gap}
To solve Problem \ref{LPFg(i)}(a) we need to find, for some position $i$ of $w$, the factor $w[1..j]^R$ with $j\leq i-g(i)$ that occurs closest to $w[1..i]$ in the lexicographically ordered list ${\mathcal L}$ of all the suffixes $w[k..n]$ of $w$ and of the mirror images $w[1..k]^R$ of its prefixes. In the following we show how to do this for all $i$ in $\bigo(n)$ time, by reducing it to answering $n$ find queries in an extended interval union-find data~structure.
\begin{theorem}\label{LPFg(i)_sol}
Problem \ref{LPFg(i)}(a) can be solved in linear time.
\end{theorem}
\begin{proof}
In a preprocessing step of our algorithm, we produce the suffix array of $w0w^R$ and the lexicographically ordered list ${\mathcal L}$ of the suffixes of $w[i..n]$ of $w$ and of the mirror images $w[1..i]^R$ of the prefixes of $w$ (which correspond to the suffixes of $w^R$). For the list ${\mathcal L}$ we compute the arrays $Rank[\cdot]$ and $Rank_R[\cdot]$. 

We first want to find, for each $i$, the prefix $w[1..j]$, such that $i-j>g(i)$, $w[1..j]^R$ occurs before $w[i..n]$ in ${\mathcal L}$ (i.e., $Rank_R[j]<Rank[i]$), and the length of the common prefix of $w[1..j]^R$ and $w[i..n]$ is greater or equal to the length of the common prefix of $w[1..j']^R$ and $w[i..n]$ for $j'$ such that $Rank_R[j']<Rank[i]$; for the prefix $w[1..j]$ as above, the length of the common prefix of $w[1..j]^R$ and $w[i..n]$ is denoted by $\LPF^{<}_{g}[i]$, while $B_<[i]$ denotes $j$. If $w[i..n]$ has no common prefix with any factor $w[1..j]^R$ with $i-j>g(i)$ and $Rank_R[j]<Rank[i]$, then  $\LPF^{<}_{g}[i]=0$ and $B_<[i]$ is not defined; as a convention, we set $B_<[i]$ to $-1$

Afterwards, we compute for each $i$ the prefix $w[1..j]$, such that $i-j>g(i)$, $w[1..j]^R$ occurs after $w[i..n]$ in ${\mathcal L}$ (i.e., $Rank_R[j]>Rank[i]$), and the length of the common prefix of $w[1..j]^R$ and $w[i..n]$ is greater or equal to the length of the common prefix of $w[1..j']^R$ and $w[i..n]$ for $j'$ such that $Rank_R[j']>Rank[i]$; for the prefix $w[1..j]$ as above, the length of the common prefix of $w[1..j]^R$ and $w[i..n]$ is denoted by $\LPF^{>}_{g}[i]$, while $B_>[i]$ denotes $j$. Clearly, $\LPF_g[i]=\max\{\LPF^{<}_g[i], \LPF^{>}_g[i]\}$. If $w[i..n]$ has no common prefix with any factor $w[1..j]^R$ with $i-j>g(i)$ and $Rank_R[j]>Rank[i]$, then  $\LPF^{>}_{g}[i]=0$ and $B_>[i]$ is set to $-1$

For simplicity, we just present an algorithm computing $\LPF^{<}_g[\cdot]$ and $B_<[\cdot]$. The computation of $\LPF^{>}_g[\cdot]$ is performed in a similar way.

The main idea behind the computation of $\LPF^{<}_g[i]$, for some $1\leq i\leq n$, is that if $w[1..j_1]$ and $w[1..j_2]$ are such that $Rank_R[j_2]<Rank_R[j_1]<Rank[i]$ and $j_1<j_2<i$ then definitely $B_<[i]\neq j_2$. Indeed, $\LCP(w[1..j_2]^R,w[i..n])<\LCP(w[1..j_1]^R,w[i..n])$, and, moreover, $i-j_2<i-j_1$. So, if $i-j_2>g(i)$ then also $i-j_1>g(i)$, and it follows that $\LPF^{<}_g[i]\geq \LCP(w[1..j_1]^R,w[i..n])> \LCP(w[1..j_2]^R,w[i..n])$, so $B_<[i]$ cannot be $j_2$. This suggests that we could try to construct for each $i$ an ordered list ${\mathcal A}_i$ of all the integers $j\leq n$ such that $Rank_R[j]<i$ and moreover, if $j_1$ and $j_2$ are in ${\mathcal A}_i$ and $j_1<j_2$ then also $Rank_R[j_1]<Rank_R[j_2]$. 

Now we describe how to implement this. Let us now consider $i_1$ and $i_2$ which occur on consecutive positions of the suffix array of $w$, such that $Rank[i_1]<Rank[i_2]$. The list ${\mathcal A}_{i_2}$ can be obtained from ${\mathcal A}_{i_1}$ as follows. We consider one by one, in the order they appear in $Rank_R[\cdot]$, the integers $j$ such that $Rank[i_1]<Rank_R[j]<Rank[i_2]$, and for each of them update a temporary list ${\mathcal A}$, which initially is equal to ${\mathcal A}_{i_1}$. When a certain $j$ is considered, we delete from the right end of the list ${\mathcal A}$ (where ${\mathcal A}$ is ordered increasingly from left to right) all the values $j'>j$; then we insert $j$ in ${\mathcal A}$. When there are no more indices $j$ that we need to consider, we set ${\mathcal A}_{i_2}$ to be equal to ${\mathcal A}$. It is clear that the list ${\mathcal A}_{i_2}$ is computed correctly. 

Now, for each $i$ we need to compute the greatest $j\in {\mathcal A}_i$ such that $j<i-g(i)$. As ${\mathcal A}_i$ is ordered increasingly, we could obtain $j$ by performing a predecessor search on ${\mathcal A}_i$ (that is, binary searching the greatest $j$ of the list, which is smaller than $i-g(i)$), immediately after we computed it, and save the answer in $B_<[i]$. However, this would be inefficient. Before proceeding, we note that if we compute the lists ${\mathcal A}_i$ for the integers $i$ in the order they appear in the suffix array of $w$, then it is clear that the time needed to compute all these lists is linear. Indeed, each $j\leq n$ is introduced exactly once in the temporary list, and then deleted exactly once from it. Doing the above mentioned binary searches would add up to a total of $\bigo(n \log n)$. We can do better than that.

Now we have reduced the original problem to a data-structures problem. We have to maintain an increasingly ordered (from left to right) list ${\mathcal A}$ of numbers (at most $n$, in the range $\{1,\ldots,2n\}$, each two different), subject to the following update and query operations. This list can be updated by the following procedure: we are given a number $j$, we delete from the right end of ${\mathcal A}$ all the numbers greater than $j$, then we append $j$ to ${\mathcal A}$. By this update, the list remains increasingly ordered.  The following queries can be asked: for a given $\ell$, which is the rightmost number of ${\mathcal A}$, smaller than $\ell$? 
We want to maintain this list while $n$ updates are executed, and $n$ queries are asked at different moments of time. 
Ideally, the total time we can spend in processing the list during all the updates should be $\bigo(n)$, and, after all the updates are processed, we should be able to provide in $\bigo(n)$ time the correct answer for all the queries (i.e., if a certain query was asked after $k$ update operations were performed on the list, we should return the answer to the query with respect to the state of the list after those $k$ update operations were completed).

Next we describe our solution to this problem.

We use a dynamic tree data-structure to maintain the different stages of ${\mathcal A}$. Initially, the tree contains only one path: the root $0$ and the leaf $2n+1$. When an update is processed, in which a number $j$ is added to the list, we go up the rightmost path of the tree (from leaf to root) until we find a node with a value smaller than $j$. Then $j$ becomes the rightmost child of that node (i.e., $j$ is a leaf). Basically, the rightmost path of a tree after $k$ updates contains the elements of ${\mathcal A}$ after those $k$ updates, preceded by $0$. When a query is asked we associate that query with the leaf corresponding to the rightmost leaf of the tree at that moment. In this way, we will be able to identify, after all updates were processed, the contents of the list at the moments of time the queries were, respectively, asked: we just have to traverse the path from the node of the tree associated to that query (this node was a leaf when the query is asked, but after all the updates were processed might have become an internal node) to the root.

The tree can be clearly constructed in linear time: each node is inserted once on the rightmost tree, and it disappears from this rightmost tree (and will not be reinserted there) when a smaller value is inserted in the tree.

In this new setting, the queries can be interpreted as weighted level ancestor queries on the nodes of the constructed tree (where the weight of an edge is the difference between the two nodes bounding it). Considering that the size of the tree is 
$\bigo(n)$, all weights are also $\bigo(n)$, there are $\bigo(n)$ queries, and these queries are to be answered off-line, it follows (see Remark~\ref{weighted_tree}) that we can return the answers to all these queries in $\bigo(n)$ time. 

This completes the linear solution to our problem. 
\end{proof}

To solve Problem \ref{LPFg(i)}(b) we use the following lemma.
\begin{lemma}\label{overlapping_LPF}
Given a word $w$, let $L[i]=\min\{j\mid j<i, \LCP(j,i)\geq \LCP(k,i)$ for all $k<i\}$. The array $L[\cdot]$ can be computed in linear time.
\end{lemma}
\begin{proof}
We first produce the suffix array of $w$. We denote by $Rank[i]$ the position of $w[i..n]$ in the suffix array of $w$. 

We first want to find, for each $i$, the suffix $w[j..n]$, such that $j<i$ is minimum with $Rank[j]<Rank[i]$ and $\LCP(j,i)\geq \LCP(j',i)$ for all $j'<i$ such that $Rank[j']<Rank[i]$.
For the suffix $w[j..n]$ as above, let $B_<[i]=j$. If no such $j$ exists, we set $B_<[i]=-1$. 

Similarly, we compute, for each $i$, the suffix $w[j..n]$, such that $j<i$ is minimum with $Rank[j]>Rank[i]$ and $\LCP(j,i)\geq \LCP(j',i)$ for all $j'<i$ such that $Rank[j']>Rank[i]$.
For the suffix $w[j..n]$ as above, let $B_>[i]=j$. Again, if no such $j$ exists, we set $B_>[i]=-1$.  

For simplicity, we just present an algorithm computing $B_<[\cdot]$. The computation of $B_>[\cdot]$ is performed in a similar way. Then, $L[i]=B_>[i]$ if $\LCP(B_>[i],i)>\LCP(B_<[i],i)$ or $\LCP(B_>[i],i)=\LCP(B_<[i],i)$ and $B_>[i]>B_<[i]$; similarly, $L[i]=B_<[i]$ if $\LCP(B_>[i],i)<\LCP(B_<[i],i)$ or $\LCP(B_>[i],i)=\LCP(B_<[i],i)$ and $B_>[i]<B_<[i]$.

The main idea behind the computation of $ B_<[i]$, for some $1\leq i\leq n$, is that if $w[j_1..n]$ and $w[j_2..n]$ are such that $Rank_R[j_2]<Rank_R[j_1]<Rank[i]$ and $j_1<j_2<i$ then definitely $B_<[i]\neq j_2$. Indeed, $\LCP(j_2,i)\leq \LCP(j_1,i)$, and, moreover, $j_2>j_1$, so $j_1$ is a better candidate than $j_2$ for $B_<[i]$. This suggests that we should construct for each $i$ an ordered list ${\mathcal A}_i$ of all the integers $j\leq n$ such that $Rank_R[j]<i$ and moreover, if $j_1$ and $j_2$ are in ${\mathcal A}_i$ and $j_1<j_2$ then also $Rank_R[j_1]<Rank_R[j_2]$. If $j'$ is on top of ${\mathcal A}_i$, then $B[i]$ is just the minimum $j$ of ${\mathcal A}_i$ such that $\LCP(j,i)=\LCP(j',i)$.

Let us now consider $i_1$ and $i_2$ which occur on consecutive positions of the suffix array of $w$, such that $Rank[i_1]<Rank[i_2]$. Assume that we already computed $B[i_1]$. The list ${\mathcal A}_{i_2}$ can be obtained from ${\mathcal A}_{i_1}$ by deleting from the right end of the list ${\mathcal A}_{i_1}$ (where ${\mathcal A}_{i_1}$ is ordered increasingly from left to right) all the values $j>i_1$. Also, to compute $B[i_2]$ more efficiently, we also delete all the values of ${\mathcal A}_{i_1}$ that are greater than $B[i_1]$. Indeed, for some $k>B[i_1]$ we have that $\LCP(B[i],i_1)=\LCP(B[i],i_2)$ and $\LCP(k,i_1)=\LCP(k,i_2)$, so $B[i_1]$ is a better candidate for $B[i_2]$ than $k$. Afterwards, we insert $i_1$ in the updated list ${\mathcal A}_{i_1}$, to obtain a list ${\mathcal A}_{i_2}$ from which $B[i_2]$ can be computed. Clearly, this entire process takes linear time, and leads to a correct computation of $B_<[\cdot]$. 

Afterwards, we compute $B_>[\cdot]$ similarly, and get $L[\cdot]$ in linear time.
\end{proof}

Another lemma shows how the computation of the array $\LPdF_g[\cdot]$ can be connected to that of the array $L[\cdot]$. For an easier presentation, let $B[i]$ denote the leftmost starting position of the longest factor $x_i$ that occurs both at position~$i$ and at a position $j$ such that $j+|x_i|\leq i-g(i)$; if there is no such factor $x_i$, then $B[i]=-1$. In other words, the length of the factor $x_i$ occurring at position $B[i]$ gives us $\LPdF_g[i]$. In fact, $\LPdF_g[i]=\min\{\LCP(B[i],i), i-g(i) -B[i]\}$. 

Now, let $L^1[i]=L[i]$ and $L^k[i]=L[L^{k-1}[i]$, for $k\geq 2$; also, we define \\
\centerline{$L^+[i]=\{L[i],L[L[i]],L[L[L[i]]],\ldots\}$.} 

The following lemma shows the important fact that $B[i]$ can be obtained just by looking at the values of $L^+[i]$. More precisely, $B[i]$ equals $L^{k}[i]$, where $k$ is obtained by looking at the values $L^j[i]\leq i-g(i)$ and taking the one such that the factor starting on it and ending on $i-g(i)-1$ has a maximal common prefix with $w[i..n]$. Afterwards, Theorem \ref{LPFg(i)_sol_rep} shows that this check can be done in linear time for all $i$, thus solving optimally Problem \ref{LPFg(i)}. 

\begin{lemma}\label{iterated_L}
For a word $w$ of length $n$ and all $1\leq i\leq n$ such that $B[i]\neq -1$, we have that $B[i]\in L^+[i]$. 
\end{lemma}
\begin{proof}
Let us assume, for the sake of a contradiction, that $B[i]=j\notin L^+[i]$. This means that $j<i-g(i)$ and $\min\{\LCP(j,i), i-g(i) -j\}\geq \min\{\LCP(j',i), i-g(i) -j'\}$ for all $j'< i-g(i)$. We further consider two simple cases.

In the first case, there is $j'<j$ such that $j'\in L^+[i]$; take $j'$ to be the greatest number less than $j$ that belongs to $L^+[i]$. Then, it is clear that the longest factor that occurs both at $j'$ and at $i$ ends before $i-g(i)-1$. Otherwise, we would have $\min\{\LCP(j,i), i-g(i) -j\}\leq i-g(i)-j < \min\{\LCP(j',i), i-g(i) -j'\}=i-g(i) -j'$, so we would have $B[i]=j'$, a contradiction. So, it follows immediately, that $\LCP(j,i)>\LCP(j',i)$. Now, if $j'=L^k[i]$ (that is $L[L[\ldots L[i]\ldots]]$, where $L$ is applied $k$ times on $i$), then there exists $k'<k$ such that $j''=L^{k'}[i]>j$ and $\LCP(j'',i)>\LCP(j',i)$ and $\LCP(j',i)=\LCP(L[j''],i)$. Clearly, this means that $j$ occurs in the suffix array of $w$ closer to the suffix $w[i..n]$ than the suffix $w[L[j'']..n]$, but farther than the suffix $w[j''..n]$. So, $\LCP(j,j'')>\LCP(L[j''],j)$, which is a contradiction to the definition of $L[j'']$. In conclusion, we cannot have that there is $j'<j$ such that $j'\in L^+[i]$.

Now, assume that there is no $j'<j$ such that $j'\in L^+[i]$. As $j\notin L^+[i]$, this leads immediately to a contradiction. Since $B[i]\neq -1$, we have that the suffix starting at position $j$ has at least the prefix of length one common with all the factors starting at positions from $L^+[i]$. Therefore, at least one element from $L^+[i]$ should be less or equal to $j$. 
This concludes our proof.
\end{proof}

\begin{theorem}\label{LPFg(i)_sol_rep}
Problem \ref{LPFg(i)}(b) can be solved in linear time.
\end{theorem}
\begin{proof}
The main idea in this proof is that, to compute $\LPdF_g[i]$, it is enough to check the elements $j\in L^+[i]$ with $j\leq i-g(i)$, and choose from them the one for which $\min\{\LCP(j,i), i-g(i) -j\}$ is maximum; this maximum will be the value we look for. In the following, we show how we can obtain these values efficiently.

First, if $j=L[i]$ for some $i$, we define the value $end[j]=k\leq \LCP(L[j],i)$, where $|w[j..k]|\leq \min\{\LCP(L[j],i), k -L[j]\}$. Basically, for each position $k'\leq k$, the longest factor $x$ starting at $L[j]$ and ending on $k'$ which also occurs at position $i$ is longer than any factor starting on position $j$ and ending on $k'$ which also occurs at position $i$. Now we note that if $j\in L^+[i]$ is the greatest element of this set such that $end[j]\leq i-g(i)-1$, then $\LPdF_g[i]=\min\{\LCP(j,i), i-g(i) - j\}$. Clearly, $end[j]$ can be computed in constant time for each $j$.

To be able to retrieve efficiently for some $i$ the greatest element of this set such that $end[j]\leq i-g(i)-1$ we proceed as follows.

First we define a disjoint-set union-find data structure on the universe $U=[1,n+1)$, where the unions can be only performed between the set containing $i$ and that containing $L[i]$, for all $i$. Initially, each number between $1$ and $n$ is a singleton set in this structure. Moreover, our structure fulfils the conditions that the efficient union-find data structure of \cite{Gabow83} should fulfil: the unions we make form a tree. 

Further, we sort in linear time the numbers $i-g(i)$, for all $i$; we also sort in linear time the numbers $end[k]$ for all $k\leq n$. We now traverse the numbers from $n$ to $1$, in decreasing order. When we reach position $j$ we check whether $j$ equals $end[k]$ for some $k$; if yes, we unite the set containing $k$ with the set containing $end[k]$ for all $k$ such that $end[k]=j$. Then, if $j=i-g(i)$ for some $i$, we just have to return the minimum of the set containing $i$; this value gives exactly 
the greatest element $j\in L^+[i]$ such that $end[j]\leq i-g(i)-1$. So, as described above, we can obtain from it the value of $\LPdF_g[i]$.  The computation of this array follows from the previous remarks.

\begin{figure}
\begin{center}
\includegraphics[width=\linewidth]{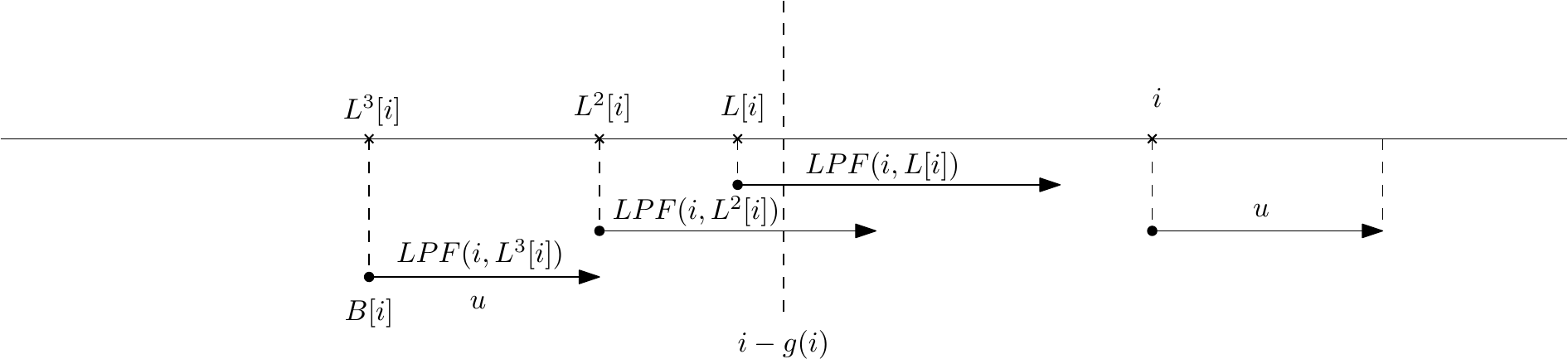}
\end{center}
\caption{Proof of Theorem \ref{LPFg(i)_sol_rep}: We measure the length of the common prefix between each factor starting on $L^k[i]$ and ending on $i-g(i)$ and  the suffix $w[i..n]$. The starting position of the factor that produces the longest such common prefix is $B[i]$. Also, this longest common prefix defines the arm $u$ of the gapped repeat $uvu$.}
\end{figure}

To evaluate the complexity of our approach, note that we do $\bigo(n)$ union operations and $\bigo(n)$ find queries on the union-find data structure. By the results of \cite{Gabow83}, the time needed to construct the union-find data structure and perform these operations on it is also $\bigo(n)$. From every find query we get in constant time the value of a element $\LPdF_g[i]$. So the solution of Problem \ref{LPFg(i)} is linear. 
\end{proof}

Just like in the previous section, the following consequence of the previous two theorems is immediate.

\begin{theorem}
Given $w$ of length $n$ and the values $g(1),\ldots,g(n)$ of $g:\{1,\ldots,n\}\ra \{1,\ldots,n\}$, we can find in linear time the gapped palindrome (or repeat) $u^Rvu$ (respectively, $uvu$) occurring in $w$, with the longest arm $u$, such that if its right arm starts on position $i$ then $|v|\geq g(i)$. 
\end{theorem}

\section{$\alpha$-gapped Repeats and Palindromes}

Recall that an $\alpha$-gapped palindrome (respectively, repeat) $w[i..j]vw[i'..j']$ is called maximal if the arms cannot be extended to the right or to the left: neither $w[i..j+1]v'w[i'-1..j']$ nor $w[i-1..j]vw[i'..j'+1]$ (respectively, neither $w[i..j+1]v'w[i'..j'+1]$ nor $w[i-1..j]v''w[i'-1..j']$) are $\alpha$-gapped palindromes (respectively, repeats).
\cite{fct} defined algorithms that find the longest $\alpha$-gapped palindromes and repeats in $\bigo(\alpha n)$ time; these algorithms do not compute the set of all $\alpha$-gapped palindromes or repeats, but just the ones with the longest arms. We present them below.

We first consider the case of $\alpha$-gapped repeats.
\begin{theorem}\label{algorithm_rep_case_aperiodic}
Given a word $w$ of length $n$ and an integer $\alpha\leq  n$, the longest $\alpha$-gapped repeat $uvu$ contained in $w$ can be found in $\bigo(\alpha n)$ time. 
\end{theorem}
\begin{proof}
Informally, our approach works as follows (see also Figure \ref{aper}). For each $k$, we try to find the longest $\alpha$-gapped repeat $u_1vu_2=uvu$, with $u_1=u_2=u$, and $2^{k+1} \log n \leq |u| \leq 2^{k+2}\log n$. In each such repeat, the right arm $u_2$ must contain a factor (called $k$-block) $z$, of length $2^k \log n$, starting on a position of the form $j2^k \log n+1$. So, we try each such factor $z$, fixing in this way a range of the input word where $u_2$ could appear. Now, $u_1$ must also contain a copy of $z$. However, it is not mandatory that this copy of $z$ occurs nicely aligned to its original occurrence; that is, the copy of $z$ does not necessarily occur on a position of the form $i \log n+1$. But, it is not hard to see that $z$ has a factor $y$ of length~$2^{k-1}\log n$, starting in its first $\log n$ positions and whose corresponding occurrence in $u_1$ starts on a position of the form $i \log n+1$. Further, we can use the fact that $u_1vu_2$ is $\alpha$-gapped and apply Lemma \ref{find_occ_range} to a suitable encoding of the input word to locate in constant time for each $y$ starting in the first $\log n$ positions of $z$ all possible occurrences of $y$ on a position of the form $i \log n+1$, occurring not more than $(8\alpha+2) |y|$ positions to the left of $z$. Intuitively, each occurrence of $y$ found in this way fixes a range where $u_1$ might occur in $w$, such that $u_1vu_2$ is $\alpha$-gapped. So, around each such occurrence of $y$ (supposedly, in the range corresponding to $u_1$) and around the $y$ from the original occurrence of $z$ we try to effectively construct the arms $u_1$ and $u_2$, respectively, and see if we really get an $\alpha$-gapped repeat. In the end, we just return the longest repeat we obtained, going through all the possible choices for $z$ and the corresponding $y$'s. We describe in the following an $\bigo(\alpha n)$ time implementation of this approach.

\begin{figure}
\begin{center}
\includegraphics[width=\linewidth]{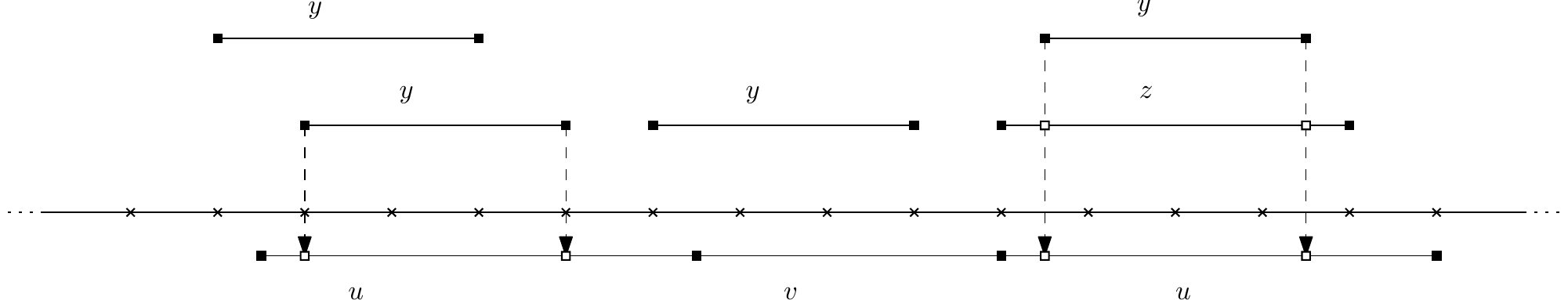}
\end{center}
\caption{
Proof of Theorem \ref{algorithm_rep_case_aperiodic}: Segment of $w$, split into blocks of length $\log n$. In this segment, $z$ is a $k$-block of length $2^k\log n$. For each factor $y$, of length $2^{k-1}\log n$, occurring in the first $\log n$ symbols of $z$ (not necessarily a sequence of blocks), we find the occurrences of $y$ that correspond to sequences of $2^{k-1}$ blocks, and start at most $(8\alpha+2)|y|=(4\alpha+1)\cdot 2^{k}\log n$ symbols (or, alternatively, $(4\alpha +1)\cdot 2^{k}$ blocks) to the left of the considered $z$. These $y$ factors may appear as runs or as separate occurrences. Some of them can be extended to form an $\alpha$-gapped repeat $u_1vu_2=uvu$ such that the respective occurrence of $y$ has the same offset in $u_1$ as the initial $y$ in $u_2$.}
\label{aper}
\end{figure} 

The first step of the algorithm is to construct a word $w'$, of length $\frac{n}{\log n}$, whose symbols, called {\em blocks}, encode $\log n $ consecutive symbols of $w$ grouped together. Basically, now we have two versions of the word $w$: the original one, and the one where it is split in blocks. Clearly, the blocks can be encoded into numbers between~$1$ and $n$ in linear time. Indeed, we produce the suffix array of $w$, and group together the suffixes that share a common prefix of length $\log n$. Then, we label (injectively) the groups obtained in this way with numbers between $1$ and $n$. Finally, a block is encoded as the label of the suffix that starts with that block. Consequently, we can construct in $\bigo(n)$ time the suffix arrays and $\LCP$-data structures for both $w$ and $w'$. We can also build in $\bigo(n)$ time the data structures of Lemma \ref{find_occ_range} for the word $w'$. 

Now, we try to find the longest $\alpha$-gapped repeat $u_1vu_2=uvu$ of $w$, with $u_1=u_2=u$, and $2^{k+1} \log n \leq |u| \leq 2^{k+2}\log n$, for each $k\geq 1$ if $\alpha>\log n$ or $k\geq \log \log n$, otherwise. Let us consider now one such $k$. We split again the word $w$, this time in factors of length $2^k \log n$, called {\em $k$-blocks}. For simplicity, assume that each split is exact. 

Clearly, if an $\alpha$-gapped repeat $u_1vu_2$ like above exists, then $u_2$ contains at least one of the $k$-blocks. Consider such a $k$-block $z$ and assume it is the leftmost $k$-block of $u_2$. On the other hand, $u_1$ contains at least $2^{k+1}-1$ consecutive blocks from $w'$, so there should be a factor $y$ of $w$ corresponding to $2^{k-1}$ of these $(2^{k+1}-1)$ blocks which is also a factor of $z$, and starts on one of the first $\log n$ positions of $z$. Now, for each $k$-block $z$ and each $y$, with $|y|=2^{k-1}\log n$ and starting in its prefix of length $\log n$, we check whether there are occurrences of $y$ in $w$ ending before~$z$ that correspond to exactly $2^{k-1}$ consecutive blocks of $w'$ (one of them should be the occurrence of $y$ in~$u_1$); note that the occurrence of $y$ in $z$ may not necessarily correspond to a group of $2^{k-1}$ consecutive blocks, but the one from $u_1$ should. As $u_1vu_2$ is $\alpha$-gapped and $|u_1|\leq 2^{k+2}\log n$, then the occurrence of $y$ from $u_1$ starts at most $(4\alpha+1)2^{k} \log n$ symbols before $z$ (as $|u_1v|\leq \alpha |u_2|\leq \alpha 2^{k+2}\log n$, and $z$ occurs with an offset of at most $2^{k}\log n$ symbols in $u_2$). So, the block-encoding of the occurrence of the factor $y$ from the left arm $u_1$ should occur in a factor of $(4\alpha+1) 2^{k}$ blocks of $w'$, to the left of the blocks corresponding to~$z$.

For the current $z$ and an $y$ as above, we check whether there exists a factor $y'$ of $w'$ whose blocks correspond to $y$, by binary searching the suffix array of $w'$ (using $\LCP$-queries on $w$ to compare the factors of $\log n$ symbols of $y$ and the blocks of $w'$, at each step of the search). If not, we try another possible $y$. If yes, using Lemma \ref{find_occ_range} for $w'$, we retrieve (in $\bigo(\log\log |w'|+\alpha)$ time) a representation of the occurrences of $y'$ in the range of $(4\alpha+1)2^{k}$ blocks of~$w'$ occurring before the blocks of $z$; this range corresponds to a range of length $(4\alpha+1)2^{k} \log n$~of~$w$. 

If $y'$ is aperiodic then there are only $\bigo(\alpha)$ such occurrences. Each factor of $w$ corresponding to one of these occurrences might be the occurrence of $y$ from $u_1$, so we try to extend both the factor corresponding to the respective occurrence of $y'$ from $w'$ and the factor $y$ from $z$ in a similar way to the left and right to see whether we obtain the longest $\alpha$-gapped repeat. If $y'$ is periodic (so, $y$ is periodic as well), we know that the representation of its occurrences consists of $\bigo(\alpha)$ separate occurrences and $\bigo(\alpha)$ runs in which $y'$ occurs (see Preliminaries). The separate occurrences are treated as above. Each run $r'$ of $w'$ where $y'$ occurs is treated differently, depending on whether its corresponding run $r$ from $w$ (made of the blocks corresponding to~$r'$) supposedly starts inside $u_1$, ends inside $u_1$, or both starts and ends inside $u_1$. We can check each of these three cases separately, each time trying to establish a correspondence between $r$ and the run containing the occurrence of $y$ from $z$, which should also start, end, or both start or end inside $u_2$, respectively. Then we define $u_1$ and $u_2$ as the longest equal factors containing these matching runs on matching positions. Hence, for each separate occurrence of $y'$ or run of such occurrences, we may find an $\alpha$-gapped repeat in $w$; we just store the longest. This whole process takes $\bigo(\alpha)$ time.

If $\alpha > \log n$, we run this algorithm for all $k\geq 1$ and find the longest $\alpha$-gapped repeat $uvu$, with $4\log n \leq |u|$, in $\bigo(\alpha n)$ time. 

If $\alpha \leq \log n$, we run this algorithm for all $k\geq \log \log n$ and find the longest $\alpha$-gapped repeat $uvu$, with $2^{\log\log n+1} \log n \leq |u|$, in $\bigo(\alpha n)$ time.  If our algorithm did not find such a repeat, we should look for $\alpha$-gapped repeats with shorter arm. Now, $|u|$ is upper bounded by $2^{\log\log n +1}\log n=2 (\log n)^2$, so $|uvu|\leq \ell_0$, for $\ell_0=\alpha\cdot 2(\log n)^2+2(\log n)^2=(2\alpha+2)(\log n)^2$. Such an $\alpha$-gapped repeat $uvu$ is, thus, contained in (at least) one factor of length $ 2\ell_0$ of $w$, starting on a position of the form $1+m\ell_0$ for $m\geq 0$.  So, we take the factors $w[1+m\ell_0..(m+2)\ell_0]$ of $w$, for $m\geq 0$, and apply for each such factor, separately, the same strategy as above. As an important detail, before running the algorithm presented above, we first encode the symbols of $w[1+m\ell_0..(m+2)\ell_0]$, which were numbers between $1$ and $n$, to numbers between $1$ and $2 \ell_0$; again, this is done by looking at the suffix array of $w$, and it allows us to apply recursively the algorithm described before. The total time needed to do that is $\bigo\left(\alpha\ell_0 \frac{n}{\ell_0}\right)=\bigo(\alpha n)$. Hence, we found the longest $\alpha$-gapped repeats $uvu$, with $2^{\log \log (2\ell_0) +1} \log (2\ell_0) \leq |u| $. If our search was still fruitless, we need to search $\alpha$-gapped repeats with $|u|\leq 2^{\log \log (2\ell_0) +1} \log (2\ell_0)\leq 16\log n$ (a rough estimation, based on the fact that $\alpha\leq \log n$). 

So, in both cases, $\alpha > \log n$ or $\alpha \leq \log n$, it is enough to find the longest $\alpha$-gapped repeats with $|u|\leq 16 \log n$. The right arm $u_2$ of such a repeat is contained in a factor $w[m\log n+1.. (m+17)\log n]$ of $w$, while $u_1$ surely occurs in a factor $x=w[m\log n-16\alpha\log n+1.. (m+17)\log n]$ (or, if $m\log n-16\alpha\log n+1\leq 0$, then in a factor $x=w[1.. (m+17)\log n]$); in total, there are $\bigo(n/ \log n)$ such $x$ factors. As before, we can process them in linear time in order to re-encode each of them as a word over the alphabet consisting of numbers which are in $\bigo(\alpha n)$. In each of these factors, we look for $\alpha$-gapped repeats $u_1vu_2=uvu$ with $2^{k+1}\leq |u|\leq 2^{k+2}$, where $0\leq k\leq \log\log n + 2$ (the case $|u|<2$ is trivial), and $u_2$ occurs in the suffix of length $17\log n$ of this factor. Moreover, $u_2$ contains a factor $y$ of the form $x[j2^k+1..(j+1)2^{k}]$. Using Lemma \ref{find_occ_small} and Remark \ref{find_occ_small_range}, 
for each such possible $y$ occurring in the suffix of length $17\log n$ of $x$, we assume it is the one contained in $u_2$ and we produce in $\bigo(\alpha)$ time a representation of the $\bigo(\alpha)$ occurrences of $y$ in the factor of length $(4\alpha+1) |y|$ preceding $y$. One of these should be the occurrence of $y$ from $u_1$. Similarly to the previous cases, we check in $\bigo(\alpha)$ time which is the longest $\alpha$-gapped repeat obtained by pairing one of these occurrences to $y$, and extending them similarly to the left and right. The time needed for this is $\bigo(\alpha \log n)$ per each of the $\bigo(\frac{n}{\log n})$ factors $x$ defined above. 
This adds up to an overall complexity of $\bigo(\alpha n)$, again.

This was the last case we needed to consider. In conclusion, we can find the longest $\alpha$-gapped repeat $uvu$ in $\bigo(\alpha n)$ time.
\end{proof}

Further we discuss the case of $\alpha$-gapped palindromes.

\begin{theorem}\label{algorithm_case_aperiodic}
Given a word $w$ of length $n$, the longest $\alpha$-gapped palindrome $u^Rvu$ contained in $w$ can be determined in $\bigo(\alpha n)$ time. 
\end{theorem}
\begin{proof} 
Our approach is similar to the case of repeats, presented in the previous theorem. For each $k$, we try to find the longest $\alpha$-gapped palindrome $u^Rvu$, with $2^{k+1} \log n \leq |u| \leq 2^{k+2}\log n$. In each such $\alpha$-gapped palindrome, the right arm $u$ must contain a factor $z$, of length $2^k \log n$, starting on a position of the form $j2^k \log n+1$. So, we try each such factor $z$, fixing in this way a range of the input word where $u$ could appear. Now, $u^R$ must contain a factor $z^R$; however, it is not mandatory that this factor  occurs at a position of the form $i \log n+1$. But, just like before, it is not hard to see that $z$ has a factor $y$, of length~$2^{k-1}\log n$, that starts in its first $\log n$ positions and whose corresponding occurrence $y^R$ from $u^R$ should start on a position of the form $i \log n+1$. 
Further, we can use the fact that $u^Rvu$ is $\alpha$-gapped and apply Lemma \ref{find_occ_range} to an encoding of the input word to locate in constant time for each $y$ starting in the first $\log n$ positions of $z$ all possible occurrences of $y^R$ on a position of the form $i \log n+1$, occurring not more than $(8\alpha+2) |y|$ positions to the left of $z$. Intuitively, each occurrence of $y$ found in this way fixes a range where $u^R$ might occur in $w$, such that $u^Rvu$ is $\alpha$-gapped. So, around each such occurrence of $y^R$ (supposedly, in the range corresponding to $u^R$) and around the $y$ from $z$ we try to effectively construct the arms $u^R$ and $u$, respectively, and see if we get the longest $\alpha$-gapped palindrome. This approach can be implemented in $\bigo(\alpha n)$ time, just like in the case of $\alpha$-gapped repeats.

The first step of the algorithm is to construct a word $w'$, of length $\frac{n}{\log n}$, whose symbols, called {\em blocks}, encode $\log n $ consecutive symbols of $w$ grouped together. Now we have two versions of the word $w$: the original one, and the one made of blocks. As before, the blocks can be encoded as numbers between $1$ and $n$ in linear time. We construct in $\bigo(n)$ time the suffix arrays and $\LCP$-data structures for both $w$ and $w'$, and we build in $\bigo(n)$ time the data structures of Lemma \ref{find_occ_range} for the word $w'$.

Considering the original word $w$, we find the $\alpha$-gapped palindrome $u^Rvu$ with $2^{k+1} \log n \leq |u| \leq 2^{k+2}\log n$, for each $k\geq \log \log n$. To this end, we split again the word $w$, this time in factors of length $2^k \log n$, called {\em $k$-blocks}. For simplicity, assume that each split we do is exact; to achieve this, we may have to pad the word with a new symbol in a suitable manner. 

If an $\alpha$-gapped palindrome $u^Rvu$ of the kind we search for exists, then $u$ contains at least one of the $k$-blocks. Consider such a $k$-block $z$ and assume it is the leftmost $k$-block of $u$. On the other hand, $u^R$ contains at least $2^{k+1}-1$ consecutive blocks from $w'$, so there should be a factor $y$ of $w$ with $y^R$ corresponding to $2^{k-1}$ of these $(2^{k+1}-1)$ blocks such that $y$ is a factor of $z$ that starts on one of its first $\log n$ positions. Now, for each $k$-block $z$ and each $y$ starting it its prefix of length $\log n$, with $|y|=2^{k-1}\log n$, we check whether there are occurrences of $y^R$ ending before~$z$ (one of them should be the occurrence of $y^R$ in $u^R$) that correspond to exactly $2^{k-1}$ consecutive blocks of $w'$. Note that the occurrence of $y$ in $z$ may not necessarily correspond to a group of $2^{k-1}$ consecutive blocks, but the one of $y^R$ from $u^R$ do. As $u^Rvu$ is $\alpha$-gapped and $|u^Rv|\leq \alpha 2^{k+2}\log n$, then the occurrence of $y^R$ from $u^R$ starts at most $(4\alpha + 1)\cdot 2^{k} \log n$ symbols before $z$. So, the block-encoding of $y^R$ should occur in a factor of  $(4\alpha +1)\cdot 2^k$ blocks of $w'$, to the left of the blocks corresponding to~$z$.

\begin{figure}\begin{center}
\includegraphics[width=\linewidth]{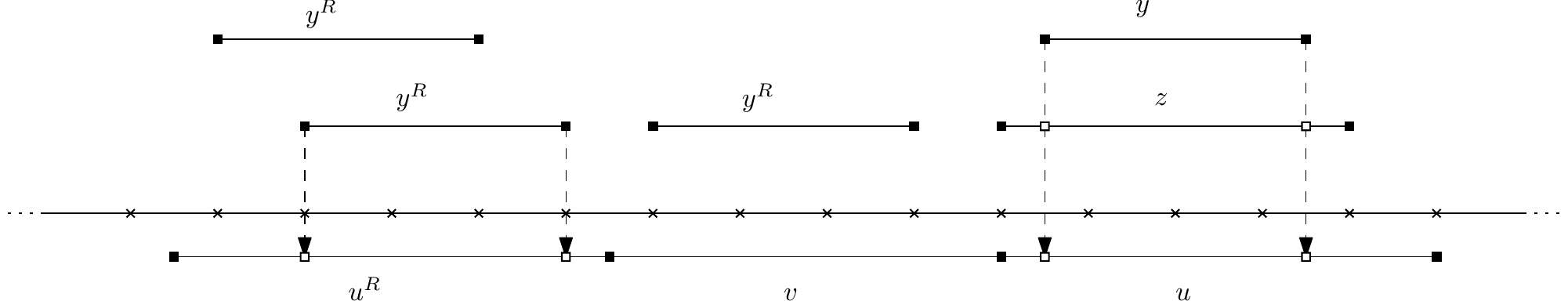}
\end{center}
\caption{
Proof of Theorem \ref{algorithm_case_aperiodic}: Segment of $w$, split into blocks of length $\log n$. In this segment, $z$ is a $k$-block of length $2^k\log n$. For each factor $y$, of length $2^{k-1}\log n$, occurring in the first $\log n$ symbols of $z$ (not necessarily a sequence of blocks), we find the occurrences of $y^R$ that correspond to sequences of $2^{k-1}$ blocks, and start at most $(8\alpha+2)|y|=(4\alpha+1)\cdot 2^{k}\log n$ symbols (or, alternatively, $(4\alpha +1)\cdot 2^{k}$ blocks) to the left of the considered $z$. These $y^R$ factors may appear as runs or as separate occurrences. Some of them can be extended to form an $\alpha$-gapped palindrome $u^Rvu$ such that the respective occurrence of $y^R$ is the mirror image of the initial $y$ in $u_2$.}
\end{figure}

For the considered $z$ and some $y$ as above, by binary searching the suffix array of $w'$ (using $\LCP$-queries on the word $ww^R$ to compare the factors of $\log n$ symbols of $y^R$ and the blocks of $w'$, at each step of the search), we check whether there exists a factor $y'$ of $w'$ whose blocks correspond to $y^R$. If not, we try another possible $y$. If yes, we continue. Using Lemma \ref{find_occ_range} for $w'$, we obtain a representation of the occurrences of $y'$ in the range of $(4\alpha+1)\cdot 2^k$ blocks of $w'$, occurring before the blocks that correspond to $z$; note that this range corresponds to a range of length $(4\alpha+1)\cdot 2^k \log n$ of $w$. 

If $y'$ is aperiodic then there are $\bigo(\alpha)$ such occurrences. Each factor corresponding to these occurrences might be the occurrence of $y^R$ from $u^R$, so we try to extend it and the $y$ from $z$ in a similar way to see whether we obtain the longest $\alpha$-gapped palindrome. 

If $y'$ is periodic (so, $y$ is periodic as well), the representation of its occurrences consists of $\bigo(\alpha)$ separate occurrences and $\bigo(\alpha)$ runs in which $y'$ occurs. The separate occurrences are treated as above. Each run $r'$ of $w'$ where $y'$ occurs is treated differently, depending on whether its corresponding run $r$ from $w$ (made of the blocks corresponding to $r'$) supposedly starts inside $u^R$, ends inside $u^R$, or both starts and ends inside $u^R$. We can check all these three cases separately, each time trying to establish a correspondence between $r$ and the run containing the occurrence of $y$ from $z$, which should also end, start, or both start or end inside $u$, respectively. Then we define $u^R$ and $u$ as the longest mirrored words containing these matching runs on mirrored positions. In this way, for each separate occurrence of $y'$ or run of such occurrences, we found a $\alpha$-gapped palindrome in $w$; we just store the longest. This whole process takes $\bigo(1)$ time for each run.

If $\alpha > \log n$, we run this algorithm for all $k\geq 1$ and find the longest $\alpha$-gapped palindrome $u^Rvu$, with $4\log n \leq |u|$, in $\bigo(\alpha n)$ time. 

If $\alpha \leq \log n$, we run this algorithm for all $k\geq \log \log n$ and find the longest $\alpha$-gapped palindrome $u^Rvu$, with $2^{\log\log n+1} \log n \leq |u|$, in $\bigo(\alpha n)$ time. If our algorithm did not find such a palindrome, we should look for $\alpha$-gapped palindrome with shorter arm. The length of this arm, $|u|$, is now upper bounded by $2^{\log\log n +1}\log n=2 (\log n)^2$, so $|u^Rvu|\leq \ell_0$, for $\ell_0=\alpha\cdot 2(\log n)^2+2(\log n)^2=(2\alpha+2)(\log n)^2$. Such an $\alpha$-gapped palindrome $u^Rvu$ is, thus, contained in (at least) one factor of length $ 2\ell_0$ of $w$, starting on a position of the form $1+m\ell_0$ for $m\geq 0$.  So, we take the factors $w[1+m\ell_0..(m+2)\ell_0]$ of $w$, for $m\geq 0$, and apply for each such factor, separately, the same strategy as above. As noted in the previous proof, these words can be re-encoded in linear time as words over an alphabet of size $\bigo(\log n)$. The total time needed to do that is $\bigo\left(\alpha\ell_0 \frac{n}{\ell_0}\right)=\bigo(\alpha n)$. Hence, we found the longest $\alpha$-gapped palindromes $u^Rvu$, with $2^{\log \log (2\ell_0) +1} \log (2\ell_0) \leq |u| $. If our search was still fruitless, we search $\alpha$-gapped palindromes with $|u|\leq 2^{\log \log (2\ell_0) +1} \log (2\ell_0)\leq 16\log n$ (a rough estimation, based on the fact that $\alpha\leq \log \log n$). 

Now in both cases (when $\alpha > \log n$ or $\alpha \leq \log n$) it is enough to find the $\alpha$-gapped palindromes with $|u|\leq 16 \log n$. The right arm $u$ of such a repeat is contained in a factor $w[m\log n+1.. (m+17)\log n]$ of $w$, while $u^R$ surely occurs in a factor $x=w[m\log n-16\alpha\log n+1.. (m+17)\log n]$ (or, if $m\log n-16\alpha\log n+1\leq 0$, then in a factor $x=w[1.. (m+17)\log n]$). In total, there are $\bigo(n/ \log n)$ such $x$ factors, and after a linear time preprocessing we can ensure that they are all over integer alphabets with respect to their length. In each of them, we look for $\alpha$-gapped palindromes $u^Rvu$ with $2^{k+1}\leq |u|\leq 2^{k+2}$, where $0\leq k\leq \log\log n + 1$ (the case $|u|<2$ is trivial), and $u_2$ occurs in the suffix of length $9\log n$ of this factor. Moreover, $u$ contains a factor $y$ of the form $x[j2^k+1..(j+1)2^{k}]$. Using Lemma \ref{find_occ_small} and Remark \ref{find_occ_small_range}, 
for each such possible $y$ occurring in the suffix of length $17\log n$ of $x$, we assume it is the one contained in $u$ and we produce in $\bigo(\alpha)$ time a representation of the $\bigo(\alpha)$ occurrences of $y^R$ in the factor of length $(4\alpha+1) |y|$ preceding $y$. One of these should be the occurrence of $y^R$ from $u^R$. Similarly to the previous cases, we check in $\bigo(\alpha)$ time which is the longest $\alpha$-gapped palindrome obtained by pairing one of these occurrences to $y$, and extending them similarly to the left and right. The time needed for this is $\bigo(\alpha \log n)$ per each of the $\bigo(\frac{n}{\log n})$ factors $x$ defined above. 
This adds up to an overall complexity of $\bigo(\alpha n)$, again.

This was the last case we needed to consider. In conclusion, we can find the longest $\alpha$-gapped palindrome $u^Rvu$ in $\bigo(\alpha n)$ time.
\end{proof}

The algorithms presented in the previous two proofs were non-trivially extended by \cite{STACS2016} to algorithms that construct the sets of all maximal $\alpha$-gapped repeats and maximal $\alpha$-gapped palindromes (which have a non-empty gap) in $\bigo(\alpha n)$ time. Essentially, instead of looking for the longest $\alpha$-gapped repeat (or palindrome) that contains a certain basic factor (as we did in this proof), we look for all the maximal $\alpha$-gapped repeats (respectively, palindromes) that contain the respective basic factor. Using a series of deep combinatorial observations on the structure of these maximal gapped repeats or palindromes, one can output all of them in $\bigo(1)$ time per repeat or palindrome. Using the crucial fact that the number of both $\alpha$-gapped repeats with non-empty gap as well as $\alpha$-gapped palindromes with non-empty gap is $\bigo(\alpha n)$ (in fact, the main result of \cite{STACS2016}), we get that they can all be identified and output in $\bigo(\alpha n)$ time.

Accordingly, we further show that given the set $S$ of all factors of a word which are maximal $\alpha$-gapped palindromes (respectively, repeats) we can compute the array $\LP$ (respectively, $\LR$) for that word in $\bigo(n+|S|)$ time. As a consequence, for constant $\alpha$, these problems can be solved in linear time. 

Note that, in this case, our strategy is fundamentally different from the ones we used in the cases of Problems \ref{LPFgG} and \ref{LPFg(i)}. There we were able to construct the desired data structures without constructing first the set of all maximal gapped palindromes and maximal gapped repeats whose gap fulfilled the required restrictions. Here we first find all maximal $\alpha$-gapped repeats and $\alpha$-gapped palindromes using the algorithms of \cite{STACS2016}, and then compute directly the desired data structures. 

\begin{theorem}\label{proof_LLAP}
Problem \ref{LLAP}(a) can be solved in $\bigo(\alpha n)$ time.
\end{theorem}
\begin{proof}
We assume that we are given an input word $w$, for which the set $S$ of all maximal $\alpha$-gapped palindromes is computed, using the algorithm of~\cite{STACS2016}. 

Let us consider a maximal $\alpha$-gapped palindrome $w[i..i+\ell-1]vw[j..j+\ell-1]$, with $w[i..i+\ell-1]^R=w[j..j+\ell-1]$. For simplicity, let us denote by $\delta=|v|=j-i-\ell$, the length of the gap; here, $i$ and $j+\ell-1$ will be called the {\em outer ends} of this palindrome, while $j$ and $i+\ell-1$ are the {\em inner ends}. 

It is not hard to see that from a maximal $\alpha$-gapped palindrome one can get a family of $\alpha$-gapped palindromes whose arms cannot be extended by appending letters simultaneously to their outer ends. We now show how this family of $\alpha$-gapped palindromes can be computed. Intuitively, we extend simultaneously the gap in both directions, decreasing in this way the length of the arms of the palindrome, until the gap becomes long enough to violate the $\alpha$-gapped restriction. The longest possible such extension of the gap can be easily computed. 

Indeed, let $r=\left \lfloor \frac{(\alpha-1)\ell-\delta}{\alpha +1}\right\rfloor$. It is not hard to check that for $r'\leq r$ we have that $w[i..i+\ell-r'-1]v'w[j+r'..j+\ell-1]$, with $v'=w[i+\ell-r'-1..i+\ell-1]vw[j..j+r'-1]$, is an $\alpha$-gapped palindrome whose left arm cannot be extended by appending letters to their outer ends. For $r'>r$ we have that $w[i..i+\ell-r'-1]v'w[j+r'..j+\ell-1]$, with $v'=w[i+\ell-r'-1..i+\ell-1]vw[j..j+r'-1]$, is still a gapped palindrome, but it is not $\alpha$-gapped anymore.  So, for a maximal $\alpha$-gapped palindrome $p=w[i..i+\ell-1]vw[j..j+\ell-1]$, we associate the interval $I_p=[j,j+r]$, and associate to it a weight $g(I_p)=j+\ell-1$. Intuitively, we know that at each position $j'\in I_p$ there exists a factor $u$, ending at position $j+\ell-1$, such that $u^Rv$ is a suffix of $w[1..j'-1]$ for some $v$. 

On the other hand, if $u$ is the longest factor starting at some position $j'\leq n$ such that $u^Rv$ is a suffix of $w[1..j-1]$, then the factor $w[i..j]=u^Rvu$ is, in fact, a maximal $\alpha$-gapped palindrome $x^Ryx$(i.e., $u^R$ is a prefix of $x^R$ and $u$ is a suffix of $x$). In other words, $u$ and $u^R$ could be extended simultaneously inside the gap, but not at the outer ends.

\begin{figure}
\begin{center}
\includegraphics[width=\linewidth]{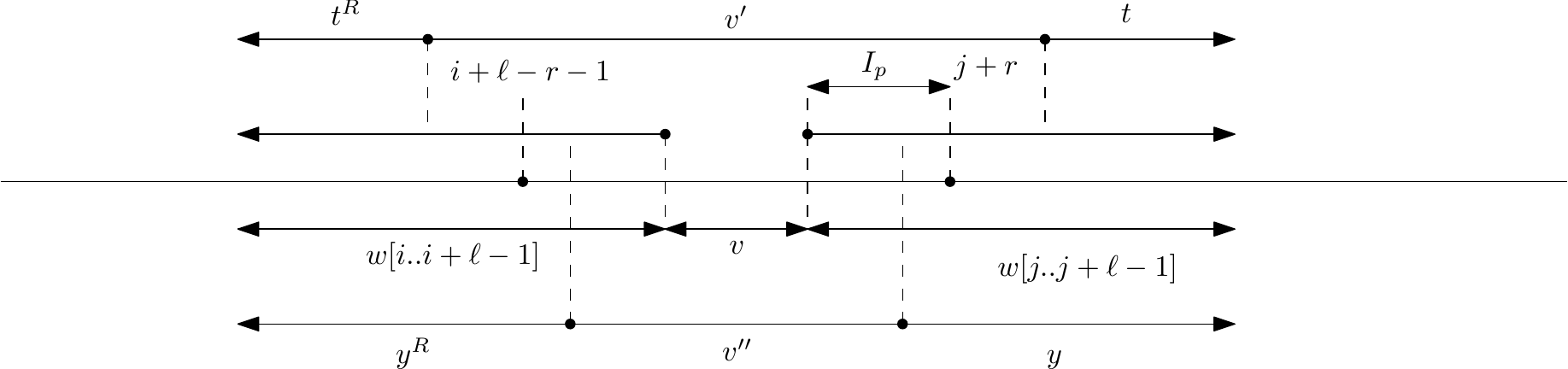}
\end{center}
\vspace{-0.5cm}
\caption{Proof of Theorem \ref{proof_LLAP}: $p=uvu$ is a maximal $\alpha$-gapped palindrome (here $\alpha=2$) with $|v|=\delta$ and $|u|=\ell$. We define $r=\left \lfloor \frac{(\alpha-1)\ell-\delta}{\alpha +1}\right\rfloor=\left \lfloor \frac{\ell-\delta}{3}\right\rfloor$. This allows us to define the interval $I_p$, where the right arm of an $\alpha$-gapped palindrome obtained from $p$ may start. Further, $y^Rv''y$ is an example of such an $\alpha$-gapped palindrome; $t^Rv't$ is a gapped palindrome obtained from $p$, which is not $\alpha$-gapped: it's right arm does not start in $I_p$.
The interval $I_p$ has weight $j+\ell-1$.}
\end{figure}
 
Consequently, to compute $\LP[i]$ for some $i\leq n$ we have to find the $\alpha$-gapped palindromes $p\in S$ for which the interval $I_p$ contains $i$. Then, we identify which of these intervals has the greatest weight. Say, for instance, that the interval $I_p$m which contains $i$, is the one that weight maximal weight $k$ from all the intervals containing $i$. Then $\LP[j]=k-j+1$. Indeed, from all the factors $u$ starting at position $j$, such that $u^Rv$ is a suffix of $w[1..j'-1]$ for some $v$, there is one that ends at position $k$, while all the other end before $k$ (otherwise, the intervals associated, respectively, to the maximal $\alpha$-gapped palindromes containing each of these factors $u^Rvu$ would have a greater weight). So, the palindrome ending at position $k$ is the longest of them all.

This allows us to design the following algorithm for the computation of $\LP[j]$. We first use the algorithm of \cite{KK09} to compute the set $S$ of all maximal $\alpha$-gapped palindromes of $w$. For each maximal $\alpha$-gapped palindrome $p=w[i..i+\ell-1]w[i+\ell..j-1]w[j..j+\ell-1]$, we associate the interval $I_p=[j,j+r]$, where $r=\left \lfloor \frac{(\alpha-1)\ell-\delta}{\alpha +1}\right\rfloor$  and $\delta=j-i-\ell$, and associate to it the weight $g(I_p)=j+\ell-1$. We process these $|S|$ intervals, with weights and bounds in $[1,n]$, in $\bigo(n+|S|)$ time as in Lemma \ref{stabbing}, to compute for each $j\leq n$ the maximal weight $H[j]$ of an interval containing $j$. Then we set $\LP[j]=H[j]-j+1$. 

The correctness of the above algorithm follows from the remarks at the beginning of this proof. Its complexity is clearly $\bigo(\alpha n)$. 
\end{proof}

The solution of Problem \ref{LLAP}(b) is very similar. 
\begin{theorem}\label{sol_LLAR}
Problem \ref{LLAP}(b) can be solved in $\bigo(\alpha n)$ time.
\end{theorem}
\begin{proof}
We first use the algorithm of \cite{STACS2016} to compute the set $S$ of all maximal $\alpha$-gapped repeats with non-empty gap of $w$. For each maximal $\alpha$-gapped repeat $p=w[i..i+\ell-1]w[i+\ell..j-1]w[j..j+\ell-1]$, we associate the interval $I_p=[j,j+r]$, where $r=\left \lfloor \frac{(\alpha-1)\ell-\delta}{\alpha}\right\rfloor$ and $\delta=j-i-\ell$, and associate to it the weight $g(I_p)=j+\ell-1$. We process these $|S|$ intervals, with weights and bounds in $[1,n]$, in $\bigo(n+|S|)$ time as in Lemma \ref{stabbing}, to compute for each $j\leq n$ the maximal weight $H[j]$ of an interval containing $j$. Now, we use Lemma \ref{centred_squares} to compute the values $SC[j]$ for each $j\leq n$. We set $\LR[j]=\max\{H[j]-j+1,SC[j]\}$.

The complexity of this algorithm is $\bigo(\alpha n)$, as $|S|\in \bigo(\alpha n)$ (see \cite{STACS2016}).

The correctness of the algorithm follows from the following remark. For a maximal $\alpha$-gapped repeat $p=w[i..i+\ell-1]w[i+\ell..j-1]w[j..j+\ell-1]$ let $r=\left \lfloor \frac{(\alpha-1)\ell-\delta}{\alpha}\right\rfloor$, where $\delta=j-i-\ell$. Then the factors $w[i+r'..i+\ell-1]w[i+\ell..j+r'-1]w[j+r'..j+\ell-1]$ are $\alpha$-gapped repeats for all $r'\leq r$, whose right arm cannot be extended anymore to the right. Moreover,  the factors $w[i+r'..i+\ell-1]w[i+\ell..j+r'-1]w[j+r'..j+\ell-1]$ are gapped repeats which are not $\alpha$-gapped for all $r'> r$. The rest of the arguments showing the soundness of our algorithm are similar to those of Theorem \ref{proof_LLAP}.
\end{proof}

\section{Future Work}

In this paper we proposed a series of algorithms that construct data structures giving detailed information on the longest gapped repeats  and palindromes occurring in a given word. There are several directions in which the work presented here can be continued.

Firstly, it seems interesting to us whether Problem \ref{LPFgG}(b) (the construction of the array $\LPdF_{g,G}[\cdot]$) can be solved faster. An intermediate problem could be to check whether we can find in linear time the longest gapped repeats with the length of the gap between a given lower bound and a given upper bound. 

Secondly, although the algorithms we propose in Theorems~\ref{algorithm_rep_case_aperiodic} and~\ref{algorithm_case_aperiodic} do not rely on computing and going through all the maximal $\alpha$-gapped repeats and palindromes when looking for the longest such structure, they have asymptotically the same complexity as the (optimal) algorithms finding all such structures. Thus, it seems natural and interesting to design algorithms finding the longest $\alpha$-gapped repeat or palindrome of a word that run faster than the algorithms we proposed here (and, in particular, than the algorithms finding all these structures). Also, it is interesting whether we can construct the data structures defined in Problem~\ref{LLAP} without producing first the list of all $\alpha$-gapped repeats and palindromes. 

Lastly, following the problems studied by \cite{GMMNT13,GMN14}, one could be interested in finding the longest gapped pseudo-repeats. More precisely, for a literal bijective anti-/morphism $f$, we want to find the longest word (or words) $u$ such that a given word $w$ contains a factor $f(u)vu$ with $|v|$ subject to different length-restrictions (e.g., between a lower and an upper bound, or shorter than $|u|$ multiplied by a factor, like in the case of $\alpha$-gapped repeats and palindromes, etc.). Such repeats and palindromes are sometimes used to formalise repeats and palindromes occurring in DNA sequences. In this setting one works with the alphabet $\{A,C,G,T\}$. When we are interested in direct repeats we may take $f$ work as a morphism and model the Watson-Crick complementarity: $f(A)=T, f(C)=G, f(G)=C,$ and $f(T)=A$;  when we are interested in inverted repeats in the genetic sequence we may take $f$ work as an antimorphism, still defined by the Watson-Crick complementarity. It is not hard to see that our algorithms can also be adapted in a straightforward manner to work in the context of such gapped pseudo-repeats.

\acknowledgements
The authors thank the anonymous referees of this paper, as well as those of the conference papers which we extend here, for their valuable remarks, suggestions, and comments, that improved the quality of this manuscript. The work of Florin Manea was supported by the DFG grant 596676.

\nocite{*}
\bibliographystyle{abbrvnat}
\bibliography{f_periodic_old}
\label{sec:biblio}

\end{document}